\documentclass[sn-basic]{sn-jnl}

\usepackage{graphicx}%
\usepackage{multirow}%
\usepackage{amsmath,amssymb,amsfonts}%
\usepackage{amsthm}%
\usepackage{mathrsfs}%
\usepackage[title]{appendix}%
\usepackage{xcolor}%
\usepackage{textcomp}%
\usepackage{manyfoot}%
\usepackage{booktabs}%
\usepackage{listings}%
\usepackage{natbib}%
\usepackage{enumitem}%
\usepackage[linesnumbered, ruled,vlined]{algorithm2e}%
\usepackage{tabularx}%
\usepackage{caption}%
\usepackage{makecell}%
\usepackage{subcaption}%

\theoremstyle{thmstyleone}%
\newtheorem{theorem}{Theorem}%
\newtheorem{proposition}[theorem]{Proposition}%

\theoremstyle{thmstyletwo}%

\theoremstyle{thmstylethree}%

\setcitestyle{aysep={}}

\SetNlSty{textbf}{}{.\hspace{0.5em}}
\renewcommand{\gets}{\mathrel{{:}{=}}}
\SetKwComment{Comment}{$\triangleright$\ }{}

\begin{document}

\title[Article Title]{A Tie-breaking based Local Search Algorithm for Stable Matching Problems}

\author{\fnm{Junyuan} \sur{Qiu}}\email{jq2388@columbia.edu}

\affil{\orgname{Columbia University}}

\abstract{The stable marriage problem with incomplete lists and ties (SMTI) and the hospitals/residents problem with ties (HRT) are important in matching theory with broad practical applications. In this paper, we introduce a tie-breaking based local search (TBLS) algorithm designed to achieve a weakly stable matching of maximum size for both the SMTI and HRT problems. TBLS begins by arbitrarily resolving all ties and iteratively refines the tie-breaking strategy by adjusting the relative order within ties based on preference ranks and the current stable matching. Additionally, we introduce TBLS-E, an equity-focused variant of TBLS, specifically designed for the SMTI problem. This variant maintains the objective of maximizing matching size, while enhancing equity through two simple modifications. In comparison with ten other approximation and local search algorithms, TBLS achieves the highest matching size, while TBLS-E exhibits the lowest sex equality cost. Significantly, TBLS-E preserves a matching size comparable to that of TBLS. Both our algorithms demonstrate faster computational speed than other local search algorithms in solving large-scale instances. Moreover, our scalability analysis shows that both algorithms maintain efficient performance as problem size increases.}


\maketitle



\section{Introduction}\label{sec:introduction}

In 1962, Gale and Shapley introduced the concept of stable matchings through two problems: the stable marriage (SM) problem and the hospitals/residents (HR) problem, with the latter initially referred to as the ``college admissions problem'' \citep{GS-1962}. In an SM instance of size $n$, $n$  men with $n$ women each rank all members of the opposite gender in a strict order of preference, with the goal to pair them in such a way that no man and woman would prefer each other over their current partner. Similarly, in an HR instance of size $n$, $n$ residents and $m$ hospitals with quotas strictly rank each other, aiming to make assignments that ensure no resident and hospital would prefer each other over their current assignments. In 2003, the concept was generalized further by the introduction of the stable $b$-matching problem, which extends the basic framework to accommodate multiple capacities \citep{Stable-b-Matching}.

A natural variant of SM and HR incorporates ties and incomplete preference lists, allowing the agents to express indifference and exclude unacceptable members. These variants are known as the stable marriage problem with incomplete lists and ties (SMTI) \citep{SMTI} and the hospitals/residents problem with ties (HRT) \citep{HRT}. Recently, the SMTI and HRT problems have received significant attention from researchers due to their relevance in diverse applications such as the child-adoption matching problem \citep{SMTI-IP-Child-adoption-1, Child-adoption-2} and the student-project allocation problem \citep{Student-Project-Allocation-1, Student-Project-Allocation-2, Student-Project-Allocation-3}. When ties are present in the preference lists, three criteria are used to define the stability of a matching: \textit{weak stability}, \textit{strong stability}, and \textit{super-stability} \citep{Three-Stability-SMTI, Three-Stability-HRT}. Among these definitions, weak stability has been the focus of extensive research \citep{HA, SMTI-IP-Child-adoption-1, Weakly-Stability-1, Weakly-Stability-2}. 

Beyond achieving stability, objectives also include maximizing matching size and ensuring fairness. For a given SM or HR problem, all stable matchings have the same size \citep{Many-Stable-Matchings-Same-Size-SM-HR}. A stable matching can be easily found using the Gale-Shapley (GS) algorithm proposed by \cite{GS-1962}. For the SMTI or HRT problem, a weakly stable matching can be obtained by applying the GS algorithm after breaking all ties. However, different methods of breaking ties generally result in stable matchings of varying sizes \citep{HRT, Break-Tie-Different-Size-SMTI}. Consequently, a natural objective is to find a maximum cardinality weakly stable matching for the SMTI and HRT problems, referred to as MAX-SMTI and MAX-HRT, respectively. These problems have been proven to be NP-hard \citep{HRT, SMTI-NP-Hard}. Equity is another criterion for matchings, especially in societal applications requiring a compromise stable matching \citep{Society-Applications-1, Society-Applications-2}. A popular way to assess fairness between the two parties is the sex equality cost \citep{Sex-Equality-Cost}. The problem of minimizing the sex equality cost in the SM problem is referred to as the sex-equal stable marriage (SESM) problem, which is NP-hard \citep{Sex-Equality-Cost-NP-Hard}. Numerous studies have investigated this problem and proposed corresponding algorithms to address it \citep{SESM-2, SESM-1, SESM-3}. However, efficiently solving large instances of MAX-SMTI and MAX-HRT remains an ongoing research topic, with limited studies striving for equity while solving the MAX-SMTI problem. 

\subsection{Our contribution}\label{subsec:our-contributions}

In this paper, we introduce a tie-breaking based local search (TBLS) algorithm to address the MAX-SMTI and MAX-HRT problems. Initially, TBLS arbitrarily breaks all ties and applies a base algorithm designed for the SM or HR problem, such as the GS algorithm, to achieve a stable matching. Our algorithm then iteratively refines the tie-breaking strategy. Rather than adjusting ties randomly, TBLS implements a well-designed adjustment to purposefully introduce a \emph{blocking pair} (BP) by increasing an agent’s rank within a tie, aimed at filling the unassigned position for this agent. To escape local optima, a disruption that randomly adjusts some ties is applied either when no adjustments are available or at a low probability. Following the tie-breaking strategy refinement process, TBLS uses a BP removal process to effectively secure a stable matching from the existing matching. Instead of re-evaluating the entire agent set, this removal process requires examining only a specific subset of agents to identify and eliminate potential BPs. If this removal process exceeds a predefined time threshold, the base algorithm is reapplied to guarantee a stable matching. This iterative refinement continues until either a perfect matching is found or the maximum number of iterations is reached.

Additionally, we propose an equity-focused variant of TBLS, called TBLS-E, aimed at finding relatively fair matchings for the MAX-SMTI problem. In this variant, an algorithm designed for the SESM problem serves as the base algorithm, and the choice of tie adjustments is strategically restricted to improve fairness while maximizing matching size.

In our experimental study, we refine the random problem generator originally proposed by \cite{Weakly-Stability-1} to generate small-sized and large-sized instances of the MAX-SMTI and MAX-HRT problems. This refined generator produces more challenging instances by introducing numerous short ties. We implemented ten additional approximation and local search algorithms for comparison. Our experimental results reveal that TBLS consistently outperforms other algorithms in achieving larger matching sizes for both the MAX-SMTI and MAX-HRT problems. Significantly, the equity-focused version, TBLS-E, achieves the lowest sex equality cost for the MAX-SMTI problem while maintaining matching sizes comparable to TBLS. Moreover, both TBLS and TBLS-E exhibit faster computational speed than other local search methods when applied to large-sized problems. Scalability tests also confirm that both algorithms scale efficiently with increasing input sizes. Our algorithms introduce a new approach to solving the MAX-SMTI and MAX-HRT problems by iteratively refining tie-breaking strategies. Furthermore, TBLS-E offers a way to obtain a relatively fair matching when solving the MAX-SMTI problem.

The remainder of this paper is organized as follows: Section \ref{sec:related-work} describes the related work, Section \ref{sec:preliminaries} provides the definitions relevant to the problem, Section \ref{sec:TBLS} details the proposed algorithms, Section \ref{sec:experiments} presents the experimental results, Section \ref{sec:conclusion} concludes the paper.

\section{Related work}\label{sec:related-work}

In recent years, the MAX-SMTI and MAX-HRT problems have gained significant attention as popular research topics in the fields of operations research. \cite{SMTI-NP-Hard} demonstrated that MAX-SMTI is NP-complete even under strict restrictions on the position and length of ties. Moreover, \cite{Break-Tie-Different-Size-SMTI} showed that the MAX-SMTI problem remains NP-hard even when the preference list of each individual is limited to a maximum length of three. Similarly, \cite{HRT} confirmed that MAX-HRT is also NP-hard. Various methods have been proposed to address these two problems.

Several polynomial-time approximation algorithms have been proposed to solve the MAX-SMTI and MAX-HRT problems. A straightforward approximation algorithm for these two problems involves applying the GS algorithm directly after resolving all ties in the given instance. which serves as a 2-approximation algorithm \citep{SMTI-NP-Hard}. Building on GS, \cite{5/3-Approximation} proposed a $\frac{5}{3}$-approximation algorithm, referred to as Kir{\'a}ly’s algorithm. Leveraging some ideas from Kir{\'a}ly’s algorithm, \cite{First-3/2-Approximation} developed the first $\frac{3}{2}$-approximation algorithm, referred to as McDermid's algorithm. Subsequently, \cite{GSM-ASBM} introduced faster and simpler $\frac{3}{2}$-approximation algorithms, referred to as GSM for solving the MAX-SMTI problem and ASBM for solving the MAX-HRT problem. Building on some concepts from GSM and ASBM, \cite{GSA2-HPA} proposed natural and local linear-time algorithms, GSA2 and HPA, for solving the MAX-SMTI and MAX-HRT problems, respectively, while maintaining the same approximation ratio. All of the above research focuses on establishing approximation guarantees, without providing any experimental evaluations. 

Experimental results are commonly reported in the literature involving constraint programming or heuristic algorithms. \cite{Weakly-Stability-1} introduced a problem generator for creating random SMTI instances and proposed a constraint programming approach to address the MAX-SMTI problem. \cite{LTIU} developed a local search algorithm called LTIU for the MAX-SMTI problem, which operates by eliminating the \emph{undominated blocking pair} (UBP). Further refining the local search process, \cite{AS} applied the adaptive search (AS)—a meta-heuristic strategy proposed by \cite{AS-Meta-Heuristic}—to remove UBPs more precisely than relying solely on a single global cost function. For MAX-SMTI problems of size 100, AS is more efficient than both LTIU and McDermid's algorithm. To efficiently address large-scale MAX-SMTI problems, \cite{MCS} designed a heuristic search algorithm, MCS, that strategically eliminates UBPs based on maximum conflicts, demonstrating superior performance and lower computational costs compared with AS and LTIU within 3,000 iterations. For solving the MAX-HRT problem, the min-conflicts algorithm (MCA) was introduced by \cite{MCA}, which utilizes preference ranks to remove UBPs. Similarly, \cite{HR} developed a heuristic repair (HR) algorithm, focused on removing UBPs based on preference ranks. Both MCA and HR outperform the LTIU method in effectiveness. \cite{HA} proposed a simple heuristic algorithm, called HA, to solve MAX-HRT by gradually forming and removing resident-hospital pairs. Most of the aforementioned literature lacks comparisons with approximation algorithms in terms of solution quality. 

Additionally, extensive research has been done on the SESM problem. A polynomial time approximation algorithm, proposed by \cite{SESM-1}, seeks to closely approximate the optimal solution. \cite{SESM-3} proposed a bidirectional local search algorithm, which simultaneously explores the solution space by conducting a forward search from the man-optimal stable matching and a backward search from the woman-optimal stable matching. At each step, the algorithm employs breakmarriage operations to generate neighboring stable matchings. \cite{PDB} innovated the first unbiased and voluntary methods for deriving an equitable stable matching within cubic time, which encompass the Late Discontent Suspension (EDS), Early Discontent Suspension (LDS), and Permanent Discontent Ban (PDB) algorithms. Few studies have considered optimizing the sex equality cost in the SMTI problem.

\section{Preliminaries}\label{sec:preliminaries}

This section provides formal definitions related to the bipartite stable $b$-matching problem, a general framework for the SMTI and HRT problems, as well as the sex equality cost. 

\subsection{Bipartite stable $b$-matching}\label{subsec:bipartitie-stable-b-matching}

An undirected bipartite graph $G=(V,E)$, preference lists $L$, and a quota function $b:V \rightarrow \mathbb{N}$ are the inputs of a bipartite stable $b$-matching problem. In the graph $G$, $V=U\cup W$, where $U$ and $W$ are disjoint sets. Each vertex in $U$ is connected to a subset of vertices in $W$, and similarly, vertices in $W$ are connected to subsets in $U$. Each agent $v\in V$ has a preference list $L_v$ consisting of agents who are adjacent to $v$. We refer to all the agents in $L_v$ as $v$'s \emph{candidates}. Additionally, the preference list of each agent may include ties. We denote $r_v(x)$ as the rank of $x$ in $v$'s preference list. If $v$ prefers $x$ and $y$ equally, this is denoted by $r_v(x) = r_v(y)$. If $v$ strictly prefer $x$ to $y$, we write $r_v(x) < r_v(y)$. Furthermore, we define $\phi_v(x)=\{y\in L_x \mid y\neq v, r_x(v)=r_x(y)\}$ as a set of agents who have the same rank as $v$ in agent $x$'s preference list.

A $b$-matching $M$ is a subset of $E$ such that each agent $v$ is connected to at most $b(v)$ edges in $M$, that is, $deg_M(v) \leq b(v)$, where $deg_M$ denotes the degree of vertex $v$ in the graph $G_M=(V,M)$. The size of $M$, defined as $|M|$, refers to the number of matched pairs. We denote $M(v)$ as the set of agents who are connected to agent $v$ in $M$, referred to as $v$'s \emph{partners}. An agent $v$ is \emph{full} if $deg_M(v)=b(v)$ and is \emph{free} if $deg_M(v) < b(v)$. The set of all free agents in $M$ is denoted by $F_M$. For each free agent $f\in F_M$, let $\xi_f^M=\{x\in L_f \mid x\notin M(f)\}$ denote the set of agents who are candidates for $f$ but are not currently matched with it in matching $M$. A matching $M$ is said to be \textit{perfect} if and only if there are no free agents.

The agent $x$ is termed as $v$'s \emph{worst partner} if $x\in M(v)$ and $r_v(x) = \max \{r_v(y) \mid y \in M(v)\}$. In a $b$-matching $M$, $(u,w)$ is a \emph{blocking pair} (BP) if and only if:
\begin{enumerate}[label=(\roman*)., itemsep=5pt, parsep=0pt, topsep=5pt]
    \item $u\in L_w$, $w\in L_u$, and $(u, w)\notin M$;
    \item $u$ either is free, or strictly prefers $w$ over its worst partner in $M$;
    \item $w$ either is free, or strictly prefers $u$ over its worst partner in $M$.
\end{enumerate}
BP $(u, w_i)$ dominates BP $(u, w_j)$ from $u$'s point of view if $u$ strictly prefers $w_i$ to $w_j$. Similarly, BP $(u_i, w)$ dominates BP $(u_j, w)$ from $w$'s point of view if $w$ strictly prefers $u_i$ to $u_j$. A BP is an \emph{undominated blocking pair} (UBP) of $v$ if there are no other BPs dominating it from $v$'s perspective. Removing UBPs may decrease the number of BPs more effectively than merely eliminating BPs \citep{HR, LTIU}. A matching $M$ is a weakly stable $b$-matching if and only if it has no BPs. In this paper, since we are focusing exclusively on weakly stable matchings, we will refer to them simply as stable matchings.

The bipartite stable $b$-matching problem provides a robust framework for understanding and addressing the SMTI and HRT problems. Specifically, an SMTI problem corresponds to a bipartite stable $b$-matching problem where $U$ represents men, $W$ represents women, and $b: V \rightarrow \{1\}$. An HRT problem corresponds to a bipartite stable $b$-matching problem where $U$ represents residents and $W$ represents hospitals. In this context, the function $b$ is defined such that $b(u) = 1$ for all $u\in U$ and $b(w) = C(w)$ for all $w \in W$, where $C(w)$ is the quota of hospital $w$. We assume, without loss of generality, that $\sum_{u\in U} b(u) = \sum_{w\in W} b(w)$ for both problems. 

\subsection{Sex equality cost}\label{subsec:sex-equality-cost}

Sex equality cost is a popular measure for achieving equity in the SM problem. This metric quantifies the gap between the preferences obtained by the two sides, defined as:
\[
    d_{SM}(M) = \Bigg|\sum_{(u,w)\in M}r_u(w) - \sum_{(u,w)\in M}r_w(u)\Bigg|
\]
where $u$ and $w$ represent a man and a woman, respectively. 

In the SMTI problem, there may be several single men and women. We denote the set $M_{paired}=\{(u,w)\in M\mid u\neq\varnothing,w\neq\varnothing\}$ as the set of matching pairs consisting of matched men and women. When calculating the sex equality cost for the SMTI problem, we focus on the matched pairs, disregarding the single individuals: 
\[
    d_{SMTI}(M) = \Bigg|\sum_{(u,w)\in M_{paired}} r_u(w) - \sum_{(u,w)\in M_{paired}} r_w(u)\Bigg|
\]

Table \ref{tab:toy-smti} shows an SMTI instance with four men and four women and their respective preference lists. In the preference lists, for example, $m_1: (w_1 \; w_3) \; w_2$ means that man $m_1$ equally prefers women $w_1$ and $w_3$, and strictly prefers both over woman $w_2$. Consider a stable matching $M = \{(m_1,w_1), (m_2,w_2), (m_3, \varnothing), (m_4, \varnothing), (w_3, \varnothing), (w_4, \varnothing)\}$. The corresponding $M_{paired}$ is $\{(m_1,w_1), (m_2,w_2)\}$. Then, the sex equality cost is calculated as follows:
\[
    d_{SMTI}(M) = \Big| \big(r_{m_1}(w_1) + r_{m_2}(w_2) \big) - \big(r_{w_1}(m_1) + r_{w_2}(m_2) \big) \Big| = \Big| (1 + 2) - (1 + 1) \Big| = 1
\]

A stable matching $M$ is considered to favor the $U$ side if $\sum_{(u,w)\in M_{paired}}r_u(w) < \sum_{(u,w)\in M_{paired}}r_w(u)$. Conversely, $M$ is considered to favor the $W$ side if $\sum_{(u,w)\in M_{paired}}r_u(w) > \sum_{(u,w)\in M_{paired}}r_w(u)$.

\begin{table}
\caption{An SMTI example of size 4}
\label{tab:toy-smti}
\begin{tabularx}{\textwidth}{>{\raggedright\arraybackslash}X c >{\raggedright\arraybackslash}X}
\toprule
Men's preference list & & Women's preference list \\
\midrule
$m_1: (w_1 \; w_3) \; w_2$ & & $w_1: m_1 \; m_3 \; m_2$ \\
$m_2: w_1 \; w_2 \; w_4$ & & $w_2: (m_2 \; m_4) \; m_1$ \\
$m_3: w_1$ & & $w_3: m_1$ \\
$m_4: w_2$ & & $w_4: m_2$ \\
\bottomrule
\end{tabularx}
\end{table}

\section{A local search algorithm for SMTI and HRT}\label{sec:TBLS}

In this section, we introduce a tie-breaking based local search algorithm, named TBLS, designed to address the MAX-SMTI and MAX-HRT problems, and its equity-focused variant, TBLS-E. The TBLS algorithm is illustrated in Algorithm \ref{TBLS}. TBLS starts to find an efficient strategy $S_{best}$ and a maximum stable matching $M_{best}$ by arbitrarily breaking all the ties and implementing a base algorithm, as outlined in lines 2-5 of Algorithm \ref{TBLS}. In each iteration, TBLS performs three steps. First, it refines the current tie-breaking strategy $S$ using the subroutine \texttt{RefineTieBreakingStrategy}. Following this refinement, it obtains a corresponding stable matching $M$ using the subroutine \texttt{ObtainStableMatching}. Lastly, if $M_{best}$ is found to be worse than $M$ according to the evaluation function $E$, then $M$ and $S$ are assigned to $M_{best}$ and $S_{best}$, respectively, as specified in lines 10-13 of Algorithm \ref{TBLS}. TBLS-E is developed by incorporating two key modifications into the original TBLS framework, the details of which are provided in Section \ref{subsec:TBLS-E}.

In each iteration, the time complexity of TBLS using GS as the base algorithm is $O(n^2)$, whereas that of TBLS-E using PDB as the base algorithm is $O(n^3)$. The overall space complexity of both algorithms is $O(n^3)$. A detailed complexity analysis is provided in Section S1 of Online Resource 1.

Next, we describe the main concepts of these two subroutines:

\texttt{RefineTieBreakingStrategy} refines the current tie-breaking strategy by adjusting the relative order within ties. These adjustments, based on preference ranks and the current stable matching, aim to improve the rank of selected free agents in the preference lists of their respective candidates, thereby increasing their chances of being matched. Each adjustment targets a free agent and one of its candidates, purposefully introducing a BP to alter the current stable matching. When this subroutine is invoked, an adjustment is selected randomly to refine the strategy. Moreover, when no adjustments are available, or with a small probability, the current tie-breaking strategy is disrupted to prevent local optima by randomly adjusting some ties.

\texttt{ObtainStableMatching} obtains a stable matching by either removing BPs or directly using the base algorithm. Given the minor updates to the preference lists during each refinement, attaining a stable matching from the existing matching by removing BPs is typically more efficient than applying the base algorithm. Moreover, this characteristic enables the examination of only a specific subset of agents to identify and eliminate potential BPs, rather than re-evaluating the entire set of agents. Therefore, to effectively obtain a stable matching after the refinement process, a BP removal process is developed to identify and eliminate UBPs. To prevent this process from becoming trapped in an endless cycle, a time threshold is established. Once this threshold is exceeded, the subroutine shifts to applying the base algorithm to secure a stable matching. 

The rest of this section is structured as follows: Section \ref{subsec:evaluation-function} discusses the evaluation function, Section \ref{subsec:refine-tie-breaking-strategy} examines the subroutine \texttt{RefineTieBreakingStrategy}, Section \ref{subsec:obtain-stable-matching} presents the subroutine \texttt{ObtainStableMatching}, Section \ref{subsec:TBLS-E} outlines two modifications for developing TBLS-E, an equity-focused variant of TBLS.

\begin{algorithm}
\caption{A Tie-breaking based Local Search Algorithm}
\label{TBLS}
\SetAlgoLined

\KwIn{\begin{minipage}[t]{\linewidth}
    \strut - An instance $I$. \\
    \strut - A small probability $p_d$. \\
    \strut - A base algorithm $A$. \\
    \strut - The maximum number of iterations $max\_iters$. \\
    \strut - Two integer $k_u$ and $k_w$. \\
    \strut - The time threshold for obtaining stable matchings $T$.
\end{minipage}
}
\KwOut{\begin{minipage}[t]{\linewidth}
    \strut - A matching $M$. \\
    \strut - A tie-breaking strategy $S$.
\end{minipage}}

\SetKwFunction{FMain}{Main}
\SetKwFunction{RTBS}{RefineTieBreakingStrategy}
\SetKwFunction{OSM}{ObtainStableMatching}
\SetKwProg{Fn}{Function}{:}{\KwRet{$M_{best},S_{best}$}}
\Fn{\FMain{$I$}}{
    $S \gets $ arbitrarily break all the ties in the instance $I$\;
    $M \gets $ apply the base algorithm $A$ to find a stable matching after breaking
    all the ties in the instance $I$ using the tie-breaking strategy $S$\;
    $S_{best} \gets S$\;
    $M_{best} \gets M$\;
    $iter \gets 0$\;
    \While{$(iter \leq max\_iters)$}{
        $S,Q_a \gets $ \RTBS{$I,M,S$}\;
        $M \gets$ \OSM{$I,M,S,Q_a$}\;
        \If{$(E(M,I) \ge E(M_{best},I))$}{
            $S_{best} \gets S$\;
            $M_{best} \gets M$\;
        }
        $iter \gets iter + 1$\;
    }
}
\end{algorithm}

\subsection{Evaluation function}\label{subsec:evaluation-function}

The evaluation function is defined by two criteria. The primary criterion is to maximize the matching size. Subsequently, among matchings of identical size, the secondary criterion gives priority to freeing up agents with longer preference lists, as these agents have a higher likelihood of being matched later to increase the matching size. Therefore, we define the evaluation function as follows:
\[
    \max \; E(M,L) = \text{size}(M)\times bigM + \sum_{v\in F_M} |L_v|\times \big( b(v) - |M(v)|\big)
\]
where $bigM$ is a large enough number to be fixed later, $|L_v|$ represents the length of agent $v$'s preference list, and $b(v) - |M(v)|$ indicates the number of unassigned positions of agent $v$.

To ensure that the evaluation scores of matchings with larger matching sizes consistently exceed those of matchings with smaller matching sizes, the parameter $bigM$ is defined as follows:
\[
    bigM = \big(\max_{u \in U} |L_u| + \max_{w \in W} |L_w|\big) \times \left(N - e_M\right)
\]
where $e_M$ denotes the estimated minimum size of the matching throughout the entire search process. To simplify parameter settings, we define $e_M = c\cdot|M_{init}|$, where $c$ represents a ratio varying between 0 and 1, and $M_{init}$ denotes the matching obtained from the random tie-breaking strategy at the beginning.

\subsection{Refining the tie-breaking strategy}\label{subsec:refine-tie-breaking-strategy}

A tie-breaking strategy, denoted as $S$, converts the original preference lists $L$, which contain ties, into tie-free (i.e., strictly ordered) preference lists $L^*$. Let $r^*_v(x)$ represent the rank of $x$ in $v$'s tie-free preference list. In the subroutine \texttt{RefineTieBreakingStrategy}, outlined in Algorithm \ref{alg:refine-tie-breaking-strategy}, the tie-breaking strategy is updated as follows:
\begin{enumerate}[itemsep=5pt, parsep=0pt, topsep=5pt]
    \item Obtain a set of candidate adjustments $R^b_M$ via the subroutine \texttt{ObtainAdjustments} (line 3).
    \item If no adjustments are available, i.e. $R^b_M$ is an empty set, or with a small probability $p_d$, then a disruption is implemented to escape local optima. The disruption involves randomly selecting $k_u$ agents from set $U$ and $k_w$ agents from set $W$, and then arbitrarily breaking all ties in their respective preference lists (lines 4-11).
    \item Otherwise, if at least one adjustment is available, i.e. $R^b_M$ is non-empty, and with a probability $1-p_d$, then an adjustment is randomly selected to refine the tie-breaking strategy. This adjustment improves a specific free agent $f$'s rank in its candidate $x$'s tie-free preference list by positioning $f$ at the highest rank among the agents who share the same rank as $f$ in $x$'s preference list, i.e., the agents in $\phi_f(x)$ (lines 12-16).
\end{enumerate}
Additionally, a set $Q_a$ is maintained to record the agents whose preference lists have been altered during the refinement process. This set will subsequently be used in the BP removal process within the subroutine \texttt{ObtainStableMatching}, which is discussed in Section \ref{subsec:obtain-stable-matching}.

\begin{algorithm}
\caption{Refine tie-breaking strategy}
\label{alg:refine-tie-breaking-strategy}
\SetAlgoLined

\KwIn{\begin{minipage}[t]{\linewidth}
    \strut - An instance $I$. \\
    \strut - A matching $M$. \\
    \strut - A tie-breaking strategy $S$. \\
    \strut - Information on the ties in the instance $\phi$. \\
    \strut - A small probability $p_d$. \\
    \strut - Two integer $k_u$ and $k_w$.
\end{minipage}
}
\KwOut{\begin{minipage}[t]{\linewidth}
    \strut - A tie-breaking strategy $S$. \\
    \strut - A set of agents whose preference lists are altered $Q_a$.
\end{minipage}
}

\SetKwFunction{FMain}{RefineTieBreakingStrategy}
\SetKwFunction{OA}{ObtainAdjustments}
\SetKwProg{Fn}{Function}{:}{\KwRet{$S,Q_a$}}
\Fn{\FMain{$I,M,S$}}{
    $Q_a \gets \emptyset$ \Comment*[r]{a set of agents whose preference lists are altered}
    $R^b_M \gets $ \OA{$I,M$}\;
    \uIf{$(\text{a small probability of } p_d \textbf{ or } R^b_M = \emptyset)$}{
        $X \gets$ randomly select $k_u$ agents from $U$ and $k_w$ agents from $W$\;
        \For{$(each\; v\in X)$}{
            $Q_a \gets Q_a \cup v$\;
            \For{$(\text{each tie in v's preference list})$}{
                Break the $tie$ arbitrarily and update $S$\; 
            }
        }
    }\Else{
        $(f,x) \gets $ randomly select one adjustment from $R^b_M$\;
        $Q_a \gets Q_a \cup x$\;
        Update $S$ to position $f$ at the highest rank in $x$'s tie-free preference list among those agents in $\phi_f(x)$\;
    }
}
\end{algorithm}

Next, we detail the process of identifying adjustments that can increase the likelihood of matching free agents and the construction of $R^b_M$. In the subroutine \texttt{ObtainAdjustments}, as illustrated in Algorithm \ref{alg:obtain-candidate-adjustments}, $R^b_M$ is generated through the following two steps:
\begin{enumerate}[itemsep=5pt, parsep=0pt, topsep=5pt]
    \item Identify a set $R_M$ of all potential adjustments (lines 4-11). Given a stable matching $M$, consider each free agent $f\in F_M$. For every agent $x\in \xi_f^M$, if the set intersection $\phi_f(x)\cap M(x)$ is non-empty, an adjustment is initiated. This adjustment aims to improve $f$'s rank in $x$'s tie-free preference list by positioning $f$ at the highest rank among those agents in $\phi_f(x)$, thereby ensuring $f$ ranks above at least one agent currently matched with $x$. This adjustment introduces $(f,x)$ as a BP and enhances $f$'s likelihood of being matched with $x$. All adjustments in $M$ are stored in $R_M$, which can be represented as $R_M=\{(f,x)\mid f\in F_M, x\in \xi_f^M, \phi_f(x)\cap M(x)\neq\varnothing\}$. The set of adjustments related to one free agent $f$ is denoted by $R_M(f) = \{a\in R_M \mid a[1]=f\}$, where $a[i]$ is the $i$-th element in adjustment $a$. Additional discussions on the relationship between our proposed tie-adjustment mechanism and the well-established breakmarriage operation are provided in Section S2 of Online Resource 1.
    \item Construct $R^b_M$ by sampling from $R_M$ (lines 12-16). One free agent may have multiple adjustments, and the number of adjustments can vary among different free agents. In some cases, agents with many unassigned positions might have relatively few adjustments, while those with fewer unassigned positions might have many adjustments. This imbalance can impact the performance of the algorithm, particularly when it is difficult to increase the overall matching size by filling the positions of those agents with fewer unassigned positions. Consequently, a balanced set of adjustments, $R^b_M$, is constructed through sampling from $R_M$. Specifically, for each free agent $f\in F_M$, we randomly select $k=\text{min}\{b(f)-|M(f)|, |R_M(f)|\}$ adjustments from $R_M(f)$ for inclusion in $R^b_M(f)$. The value of $k$ is chosen to ensure that the number of candidate adjustments for each free agent $f$ does not exceed its unfilled capacity, which is defined as $b(f)-|M(f)|$.
\end{enumerate}

\begin{algorithm}
\caption{Obtain a set of candidate adjustments}
\label{alg:obtain-candidate-adjustments}
\SetAlgoLined

\KwIn{\begin{minipage}[t]{\linewidth}
    \strut - An instance $I$. \\
    \strut - A matching $M$. \\
    \strut - Information on the ties in the instance $\phi$.
\end{minipage}
}
\KwOut{A set of candidate adjustments $R^b_M$}

\SetKwFunction{FMain}{ObtainAdjustments}
\SetKwProg{Fn}{Function}{:}{\KwRet{$R^b_M$}}
\Fn{\FMain{$I,M$}}{
    $F_M \gets \text{all free agents in } M$\;
    $R_M, R^b_M \gets \emptyset$\;
    \For{$(each\; f\in F_M)$}{
        $\xi_f^M \gets \text{the set of currently unmatched candidates for } f$\;
        \For{$(each\; x\in \xi_f^M)$}{
            \If{$(\phi_f(x)\cap M(x)\neq\emptyset)$}{
                $R_M \gets R_M\cup (f,x)$\;
            }
        }
    }
    \For{$(each\; f\in F_M)$}{
        $k \gets \text{min}\{b(f)-|M(f)|, |R_M(f)|\}$\;
        $X \gets \text{randomly select } k \text{ adjustments from } R_M(f)$\;
        $R^b_M \gets R^b_M\cup X$\;
    }
}
\end{algorithm}

We take the SMTI instance described in Table \ref{tab:toy-smti} as an example to illustrate the effectiveness of adjustments in achieving a larger stable matching. Consider a tie-breaking strategy $S_1$ that assigns $r^*_{m_1}(w_1) < r^*_{m_1}(w_3)$ and $r^*_{w_2}(m_2) < r^*_{w_2}(m_4)$. The resulting stable matching $M_1=\{(m_1,w_1), (m_2,w_2), (m_3, \varnothing), (m_4, \varnothing), (w_3, \varnothing), (w_4, \varnothing)\}$ has a size of two. The balanced set of the adjustments is $R^b_{M_1}=\{(m_4, w_2), (w_3, m_1)\}$. After implementing adjustment $(m_4, w_2)$, the new tie-breaking strategy $S_2$ results in $r^*_{w_2}(m_2) > r^*_{w_2}(m_4)$. The corresponding stable matching $M_2 = \{(m_1, w_1), (m_2, w_4), (m_4, w_2), (m_3, \varnothing), (w_3, \varnothing)\}$ has a size of three. The balanced adjustments now are $R^b_{M_2}=\{(w_3, m_1)\}$. Upon performing adjustment $(w_3, m_1)$, the refined tie-breaking strategy $S_3$ assigns $r^*_{m_1}(w_1) > r^*_{m_1}(w_3)$. Consequently, the stable matching $M_3 = \{(m_1, w_3), (m_2, w_4),  (m_3, w_1), (m_4, w_2)\}$ results in a perfect matching with a size of four.

\subsection{Obtaining a stable matching}\label{subsec:obtain-stable-matching}

Following the refinement of the tie-breaking strategy, it is necessary to obtain a corresponding stable matching to assess the efficacy of the current strategy. As the preference lists of only a subset of agents are modified during the refinement process, most matching pairs from the previous stable matching can be retained. Rather than reconstructing a stable matching from scratch using the base algorithm, eliminating BPs from the existing matching is generally more efficient. Moreover, rather than re-evaluating the entire set of agents, only a specific subset of agents needs to be examined to identify and eliminate potential BPs. Therefore, we design a BP removal process to effectively identify and eliminate UBPs. However, this process may become trapped in an endless cycle while attempting to obtain a stable matching by removing BPs \citep{Endless-Cycle}. A time threshold $T$ is established for the BP removal process. If the duration of this process exceeds this threshold, the base algorithm is applied to ensure a stable matching is achieved. The function to obtain a stable matching is described in Algorithm \ref{alg:obtain-stable-matching}. The rest of this section will focus on the BP removal process, where only a specific set of agents is required for examination.\\

\begin{algorithm}
\caption{Obtain stable matching}
\label{alg:obtain-stable-matching}
\SetAlgoLined

\KwIn{\begin{minipage}[t]{\linewidth}
    \strut - An instance $I$. \\
    \strut - A matching $M$. \\
    \strut - A tie-breaking strategy $S$. \\
    \strut - A set of agents whose preference lists are altered $Q_a$. \\
    \strut - A base algorithm $A$. \\
    \strut - The time threshold for obtaining stable matchings $T$. 
\end{minipage}
}
\KwOut{A matching $M$.}

\SetKwFunction{FMain}{ObtainStableMatching}
\SetKwProg{Fn}{Function}{:}{\KwRet{$M$}}
\Fn{\FMain{$I,M,S,Q_a$}}{
    $X \gets Q_a$ \Comment*[r]{$X$ is the set of agents requiring examination}
    \While{$(execution\; time \leq T \textbf{ and } X \neq \emptyset)$}{
        $v \gets $ randomly pop one agent from $X$\; 
        $y_{worst} \gets $ the worst partner of $v$ in $M$ under $S$\;
        sort $v$'s tie-free preference list in ascending order\;
        \For{$(each\; y \in v's\; \text{tie-free preference list})$}{
            \uIf{$(y \in M(v))$}{
                continue\;
            }
            \uElseIf{$(\text{v is full} \textbf{ and } r^*_{v}(y) > r^*_{v}(y_{worst}))$}{
                break\;
            }
            \ElseIf{$((v,y)\; \text{is a blocking pair under }S)$}{
                $z_{worst} \gets $ the worst partner of $y$ in $M$ under $S$\;
                \If{$(v\; is\; full \textbf{ and } y_{worst}\; is\; full)$}{
                    $X \gets X\cup y_{worst}$\;
                }
                \If{$(y\; is\; full \textbf{ and } z_{worst}\; is\; full)$}{
                    $X \gets X\cup z_{worst}$\;
                }
                Remove BP $(v,y)$ in $M$\;
                $y_{worst} \gets $ the worst partner of $v$ in $M$ under $S$\;
            }
        }
    }
    \If{$(X \neq \emptyset)$}{
        $M \gets$ apply the base algorithm $A$ to find a stable matching after breaking all the ties in the instance $I$ using the tie-breaking strategy $S$\;
    }
}
\end{algorithm}

\begin{proposition}
Given a tie-breaking strategy $S_1$ and its corresponding stable matching $M_1$. If, after refining the tie-breaking strategy from $S_1$ to $S_2$, there exist no BPs for the agents in set $Q_a$ within the current matching $M_1$, then $M_1$ is also a stable matching corresponding to $S_2$.
\end{proposition}

\begin{proof}
Assume, for contradiction, that $M_1$ is not stable under $S_2$. Then, there exists a BP $(u,w)$ under $S_2$. We consider two cases:
\begin{enumerate}[label=(\roman*)., itemsep=5pt, parsep=0pt, topsep=5pt]
    \item $u\in Q_a \textbf{ or } w\in Q_a:$ Given the conditions, no BPs exist for the agents in $Q_a$ within $M_1$ under $S_2$. Thus, this scenario is not possible.
    \item $u\notin Q_a \textbf{ and } w\notin Q_a:$ Since $M_1$ was stable under $S_1$, and the only change from $S_1$ to $S_2$ involves the agents in $Q_a$, the stability conditions for agents outside $Q_a$ remain unaffected. Thus, a BP under $S_2$ for agents not in $Q_a$ cannot exist. 
\end{enumerate} 
As both scenarios lead to contradictions, the initial assumption that $M_1$ is not stable under $S_2$ must be incorrect. Therefore, $M_1$ is a stable matching under $S_2$ provided that no BPs exist for the agents in set $Q_a$.
\end{proof}

Given a matching $M$, to remove a BP $(u, w)$, we follow these steps. First, if $u$ is full, we disconnect $u$ from its worst partner in $M$. Similarly, if $w$ is full, we disconnect $w$ from its worst partner in $M$. Then, we connect $u$ and $w$.\\

\begin{proposition}
    Let $M_1$ be a matching with a set $B_1$ of BPs. Removing a BP $(u,w)\in B_1$ results in a new matching $M_2$ and a new set $B_2$ of BPs. Let $w'$ and $u'$ be the agents disconnected from $u$ and $w$, respectively, during this removal process. If no agent is disconnected from $u$, then $w'$ is undefined; similarly, if no agent is disconnected from $w$, then $u'$ is undefined. If $B_2 \setminus B_1 \neq \varnothing$, then at least one of $u'$ or $w'$ exists, and all BPs in $B_2 \setminus B_1$ must involve either $u'$ or $w'$.
\end{proposition}

\begin{proof}
Consider any BP $(x,y)$ that exists in $B_2$ but not in $B_1$. Since $(x,y)$ forms a BP in $M_2$ and not in $M_1$, the matchings of $x$ or $y$ must have been changed. The only changes in the matchings from $M_1$ to $M_2$ are attributed to the removal of the BP $(u, w)$, affecting the matchings of $u, w, u'$, and $w'$.
\begin{enumerate}[label=(\roman*)., itemsep=5pt, parsep=0pt, topsep=5pt]
    \item Consider agents $u$ and $w$. If $u$ was free in $M_1$, then no new BPs involving $u$ can be introduced in $M_2$ as such BPs would have been identified in $M_1$. Conversely, if $u$ was full in $M_1$, any BPs involving $u$ with candidates worse than $w$ in $B_1$ will no longer exist, and no additional BPs including $u$ will be introduced. The same logic applies to $w$.
    \item Consider agents $u'$ and $w'$. If $u'$ and $w'$ are both undefined, then clearly no new BPs involving them can arise. If either $u'$ or $w'$ exists, new BPs might be introduced. Specifically, if $w'$ was free in $M_1$, then no new BPs involving $w'$ are introduced in $M_2$, except potentially $(w',u')$, as any such BPs would have existed in $M_1$. This condition holds unless $u'$, previously full in $M_1$, becomes free in $M_2$, and if $u'$ is a candidate of $w'$, thus establishing $(w', u')$ as a new BP. If $w'$ was full in $M_1$, let $z$ represent the worst partner of $w'$ in $M_1$; in $M_2$, where $w'$ is free, possible new BPs might arise involving $w'$ and candidates worse than $z$. The same logic applies to $u'$.
\end{enumerate}
Therefore, if $B_2 \setminus B_1$ is non-empty, it necessitates that at least one of $u'$ or $w'$ exists, and all BPs in $B_2 \setminus B_1$ must involve either $u'$ or $w'$.
\end{proof}

Given these propositions, verifying the existence of BPs for each agent after refining the tie-breaking strategy is unnecessary. It is sufficient to focus on a specific subset of agents. Consequently, we maintain a set of agents requiring examination, initially set to $Q_a$. In each iteration, an agent $v$ is randomly selected from this set, and all BPs involving $v$ are removed. Specifically, if BPs involving $v$ are identified, the UBP from $v$'s perspective is removed. Subsequently, any agents who were previously full and disconnected are added to the set. The BP removal process is repeated until the set is empty or the time threshold is reached, as detailed in lines 2-24 of Algorithm \ref{alg:obtain-stable-matching}.

\subsection{Equity-focused variant of TBLS}\label{subsec:TBLS-E}

To achieve a relatively equitable outcome with a low sex equality cost when solving the MAX-SMTI problem, two modifications are introduced to TBLS:
\begin{enumerate}[label=(\roman*)., itemsep=5pt, parsep=0pt, topsep=5pt]
    \item Use an algorithm designed for the SESM problem, such as PDB, as the base algorithm.
    \item Restrict the choice of tie adjustments based on the current bias of the matching. If the matching favors the $U$ side, adjustments related to agents of $U$ are permitted. Specifically, only adjustments within the set $\{a\in R^b_M\mid a[1]\in U\}$ are allowed. Conversely, if the current matching favors the $W$ side, only adjustments related to agents of $W$ are permitted. In cases where no adjustments are available within the constraint, the restriction is temporarily lifted.
\end{enumerate}
By balancing the satisfaction across both sides, these two modifications aim to reduce the sex equality cost while maximizing the matching size. The first modification finds a stable matching with a low sex equality cost after applying the tie-breaking strategy. The second modification limits the choice of adjustments to reduce bias toward any one side during the refinement process. We refer to this variant as TBLS-E. 

To demonstrate the effectiveness of the second modification, consider the following scenario: suppose the current matching favors the $U$ side, and there is an adjustment $a=(f,x)$, where $f\in U$ and $x\in W$. If $f$ is successfully matched with $x$ following adjustment $a$, then an agent $q\in U$ would be disconnected from $x$. This disconnection may lead $q$ to form a new pairing with another agent $p$, where $r_{q}(p) \ge r_{q}(x)$. Therefore, adjustment $a$ contributes to mitigating the bias toward $U$.

\section{Experiments}\label{sec:experiments}

In this section, we compare the performance of our TBLS algorithm and its derived version, TBLS-E, with other methods in solving the MAX-SMTI and MAX-HRT problems. The experiments are performed using Python on a machine with an Intel Core i9-13900HX CPU (5.40 GHz) and 32GB of RAM. The source code is available at: \href{https://github.com/Junyuan-Qiu/Stable-Matching-Local-Search}{https://github.com/Junyuan-Qiu/Stable-Matching-Local-Search}. The rest of this section is structured as follows: Section \ref{subsec:experiments-problem-set} introduces the problem set, Section \ref{subsec:experiments-parameter-settings} details the parameter settings of our algorithms, Section \ref{subsec:experiments-performance-comparison} presents the performance comparison.

\subsection{Problem set}\label{subsec:experiments-problem-set}

We modified the random problem generator initially proposed by \cite{Weakly-Stability-1} to create more challenging instances. The updated SMTI generator accepts four parameters: size ($n$), probability of incompleteness ($p_1$), probability of initiating a tie ($p_2$), and a tie length generator ($g$). Except for tie-related processes, all generation steps remain the same as those of the original generator. The updated tie generation steps are as follows: 
\begin{enumerate}[itemsep=5pt, parsep=0pt, topsep=5pt]
    \item For an agent $v$, we begin by generating a random number $0\leq p < 1$ starting from their first choice.
    \item If $p \leq p_2$, then we use $g$ to generate a random number $i$. The next $i$ agents  are assigned the same rank as the first choice. The agent following these $i+1$ agents is assigned a rank one higher than the first choice, and a new random number $p$ is generated to continue the process from this agent.
    \item Otherwise, the next agent is assigned a rank one higher than the first choice, and the process continues in the same manner.
\end{enumerate}
When we set $g=\text{Geom}(1-p_2)$, where $\text{Geom}(1-p_2)$ denotes a geometric distribution with the parameter $1-p_2$, the instances generated by the updated generator are nearly identical to those generated by the original one. In this setting, lower values of $p_2$ result in fewer, shorter ties, while higher values lead to a greater number of longer ties. This characteristic simplifies the problem even if $p_2$ is high, as nearly all preferences are tied, indicating general indifference. Therefore, we introduce $g=\text{Geom}(p_2)$ to generate an inverse relationship between the number of ties and the length of each tie. This modification increases the problem's complexity when $p_2$ is high, as a greater number of shorter ties complicates the solving process. 

Similarly, the modified HRT generator accepts five parameters: the number of residents ($n$), the number of hospitals ($m$), the probability of incompleteness ($p_1$), the probability of starting a tie ($p_2$), and the random generator of the tie length ($g$). Experimental results indicate that HRT problems with capacities uniformly distributed among hospitals present a greater challenge compared to those with randomly allocated capacities \citep{HR}. Therefore, the capacity is uniformly distributed among hospitals in the HRT instances. 

To evaluate the efficacy of various algorithms in addressing the MAX-SMTI and MAX-HRT problems, we categorize the problems into two sizes: small and large. For the small-sized problems, we select values of $p_1$ that result in instances which are difficult to solve by simply applying the GS algorithm with randomly broken ties. For the large-sized problems, we set a high value for $p_1$ to generate more challenging instances. We systematically vary parameters across both problem sizes, with the configurations detailed in Table \ref{tab:problem-configurations}. Furthermore, we conduct scalability experiments to assess how the performance of our algorithms scales as the input size increases from 1,000 to 10,000, with the configurations provided in Table \ref{tab:problem-configurations-scalability}. In these two configuration tables, parameters are presented with specific ranges, intervals, or formulas. For instance, the notation ``0.7--0.9 (0.1)'' indicates values ranging from 0.7 to 0.9 in increments of 0.1. The use of commas, such as in ``$\text{Geom}(p_2), \text{Geom}(1-p_2)$'', denotes an alternative parameter set being considered for the parameter. Additionally, the formula ``$(n-10)/n$'' represents a value that is computed based on $n$. We define each unique combination of $n,m,p_1,p_2$, $g$ as a \textit{configuration}. To ensure a robust evaluation, we generate 100 instances for each configuration. The randomly generated instances used in this study are available in the Science Data Bank repository, \href{https://www.scidb.cn/en/s/ZfauMn}{https://www.scidb.cn/en/s/ZfauMn}.

\begin{table}
\caption{Problem configurations for the MAX-SMTI and MAX-HRT problems}
\label{tab:problem-configurations}
\begin{tabular*}{\textwidth}{@{\extracolsep\fill}lcccc}
\toprule%
& \multicolumn{2}{@{}c@{}}{MAX-SMTI} & \multicolumn{2}{@{}c@{}}{MAX-HRT} \\\cmidrule{2-3}\cmidrule{4-5}%
Parameter & Small-sized & Large-sized & Small-sized & Large-sized \\
\midrule
$n$ & 100 & 1,000 & 100 & 1,000 \\
$m$ & - & - & 10 & 10--50 (10) \\
$p_1$ & 0.7--0.9 (0.1) & 0.95--0.99 (0.01) & 0.3--0.9 (0.1) & 0.9 \\
$p_2$ & 0.1--1 (0.1) & 0.1--1 (0.1) & 0.1--1 (0.1) & 0.1--1 (0.1) \\
$g$ & \makecell{$\text{Geom}(p_2),$ \\ $\text{Geom}(1-p_2)$} & \makecell{$\text{Geom}(p_2),$ \\ $\text{Geom}(1-p_2)$} & \makecell{$\text{Geom}(p_2),$ \\ $\text{Geom}(1-p_2)$} & \makecell{$\text{Geom}(p_2),$ \\ $\text{Geom}(1-p_2)$} \\
\botrule
\end{tabular*}
\end{table}

\begin{table}
\caption{Problem configurations for the scalability experiments}
\label{tab:problem-configurations-scalability}
\begin{tabular*}{\textwidth}{@{\extracolsep\fill}lcc}
\toprule
Parameter & MAX-SMTI & MAX-HRT\\
\midrule
$n$ & 1,000-10,000 (1,000) & 1,000-10,000 (1,000) \\
$m$ & - & 30 \\
$p_1$ & $(n-10)/n$ & 0.9 \\
$p_2$ & 0.5 & 0.5 \\
$g$ &$\text{Geom}(p_2)$ & $\text{Geom}(p_2)$ \\
\botrule
\end{tabular*}
\end{table}

\subsection{Parameter settings}\label{subsec:experiments-parameter-settings}

In both TBLS and TBLS-E, the parameters $c=0.9$ and $p_d=0.05$ are consistently maintained across all configurations. The parameter $c$ is utilized to compute $e_M$, which estimates the minimum size of the matching during the search process. Meanwhile, $p_d$ denotes the probability of executing a disruption. Furthermore, the time threshold $T$ for obtaining stable matchings is set equal to the time required to obtain a stable matching with the use of the base algorithm at the beginning. The parameters $k_u$ and $k_w$, which are adapted according to the scale and complexity of the problem, are detailed in Table \ref{tab:TBLS-variable-parameters}. Specifically, $k_u$ and $k_w$ represent the number of agents selected from sets $U$ and $W$ for disruption, respectively. In the context of the MAX-SMTI problem, $U$ and $W$ represent men and women, respectively, while in the MAX-HRT problem, $U$ and $W$ correspond to residents and hospitals, respectively. For TBLS, GS serves as the base algorithm, whereas for TBLS-E, PDB is used as the base algorithm.

\begin{table}
\caption{Variable parameters for the TBLS and TBLS-E algorithms}
\label{tab:TBLS-variable-parameters}
\begin{tabular*}{\textwidth}{@{\extracolsep\fill}lcccc}
\toprule%
& \multicolumn{2}{@{}c@{}}{MAX-SMTI} & \multicolumn{2}{@{}c@{}}{MAX-HRT} \\\cmidrule{2-3}\cmidrule{4-5}%
Parameter & Small-sized & Large-sized & Small-sized & Large-sized \\
\midrule
$k_u$ & 1 & 5 & 1 & 5 \\
$k_w$ & 1 & 5 & 1 & 1 \\ 
\botrule
\end{tabular*}
\end{table}

\subsection{Performance comparison}\label{subsec:experiments-performance-comparison}

In each configuration, we calculate the average result across 100 instances for each algorithm. Specifically, for each algorithm and under each configuration, the average performance metrics reported in this study are computed over 100 instances using the following equations:
\begin{equation}
\label{eq:avg_nos}
    \text{Average number of singles / unassigned positions} = n - \frac{1}{|S|}\sum_{M_i\in S}|M_i|
\end{equation}
\begin{equation}
    \text{Average SECost} = \frac{1}{|S|}\sum_{M_i\in S}\text{SECost}(M_i)
\end{equation}
\begin{equation}
    \text{Average execution time} = \frac{1}{|I|}\sum_{i=1}^{|I|}\text{time}(i)
\end{equation}
where $I$ is the set of instances in the configuration, $M_i$ is the resulting matching produced by the algorithm when solving instance $i$, and $S$ denote the set of stable matchings among these, i.e., $S=\{M_i\mid M_i \text{ is stable for }i=1,2,\cdots,|I|\}$. 

Note that all local search algorithms used to compare the quality of the solution produced stable matches on all small-sized and large-sized instances, except that MCA solved only 99 of 100 instances in one configuration. For this reason, we did not report the percentage of stable matchings in the results.

An algorithm is deemed the ``winner'' in a configuration if its result is not worse than that of any other algorithm. To compare the performance of different algorithms, we count the number of wins and compute the overall average result across all configurations. Calculating the overall average result of an algorithm involves two steps. First, the result for each configuration is calculated by averaging the outcomes of 100 instances. Then, the results for each configuration are averaged to obtain the final overall result. These two metrics---the number of wins and the overall average---offer a comprehensive basis for comparing the performance of the algorithms across different configurations.

For the MAX-SMTI problem, we evaluate the solution quality of TBLS and TBLS-E by examining the matching size and the sex equality cost (SECost). These results are compared against those of AS, GSA2, GSM, LTIU, MCS, and PDB. Specifically, LTIU is built upon the SML2 algorithm detailed in \cite{LTIU}, and PDB is applied to obtain a stable matching after randomly breaking all the ties. Additionally, we compare the execution time of TBLS and TBLS-E with that of AS, LTIU, and MCS, as they are all heuristic search methods. For the MAX-HRT problem, we assess the solution quality of TBLS, focusing on the matching size, in comparison with ASBM, HA, HPA, HR, and MCA. Furthermore, we compare the execution time of TBLS with that of HR and MCA, since they all employ heuristic techniques. For both problems, we do not include approximation algorithms such as GSA2, HA, etc., in the execution time comparison, since they operate under fundamentally different assumptions and performance objectives compared with heuristic algorithms. The detailed results for each configuration presented in Sections \ref{subsubsec:SMTI} and \ref{subsubsec:HRT} are provided in Online Resource 2, whereas the results corresponding to Section \ref{subsubsec:scalability} are included in Online Resource 3.

\subsubsection{MAX-SMTI problems}\label{subsubsec:SMTI}

First, we compare the solution quality of different algorithms in solving small-sized MAX-SMTI problems, including 60 configurations. The maximum number of iterations is set to 3,000 for TBLS and TBLS-E, and 50,000 for LTIU, AS, and MCS, to fully leverage the capabilities of the latter three local search algorithms.

\begin{figure}
    \centering
    \begin{subfigure}[b]{0.495\textwidth}
        \centering
        \includegraphics[width=\textwidth]{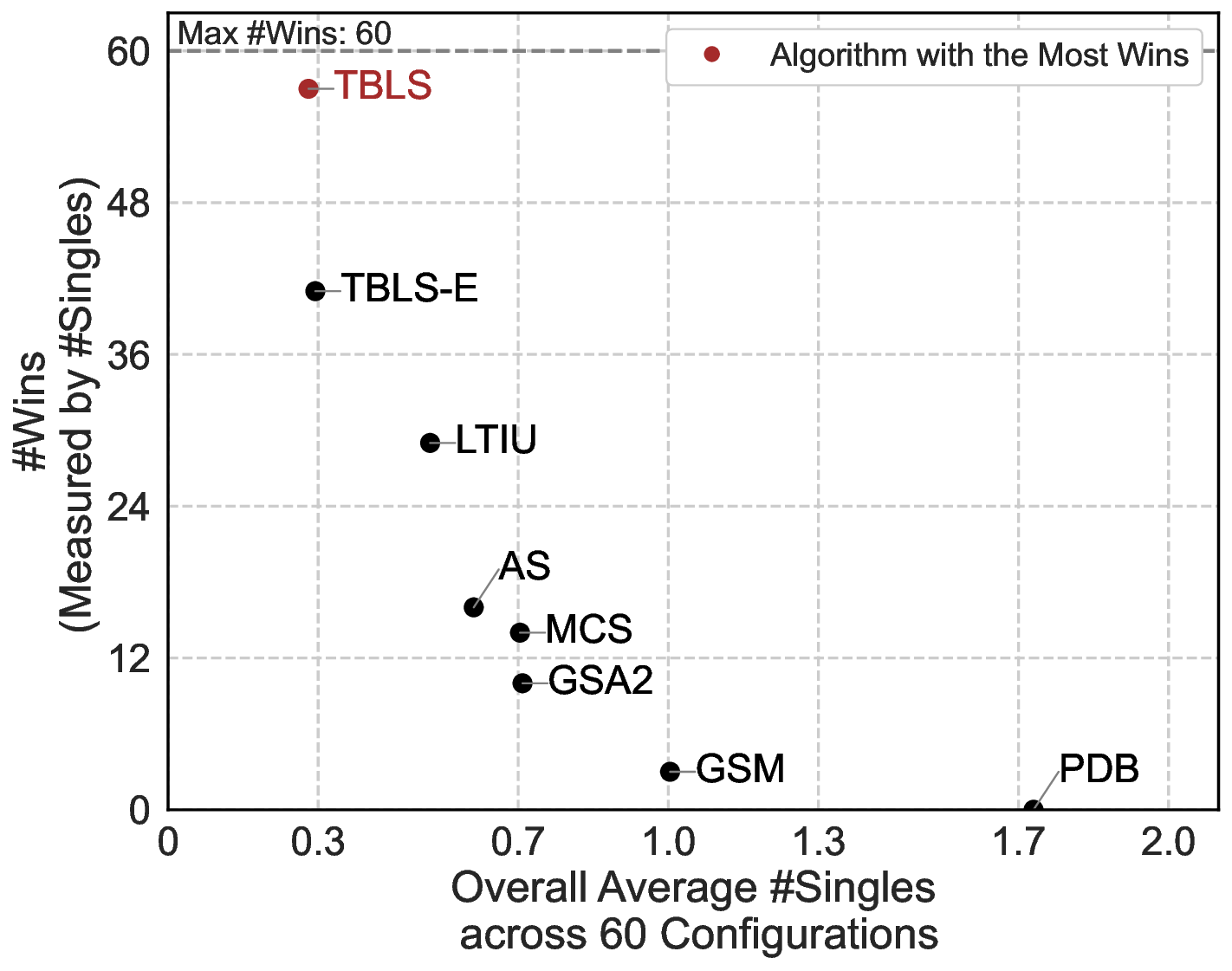}
        \caption{\textnormal{Number of singles}}
        \label{fig:small-sized-SMTI-quality-nos}
    \end{subfigure}
    \hfill
    \begin{subfigure}[b]{0.495\textwidth}
        \centering
        \includegraphics[width=\textwidth]{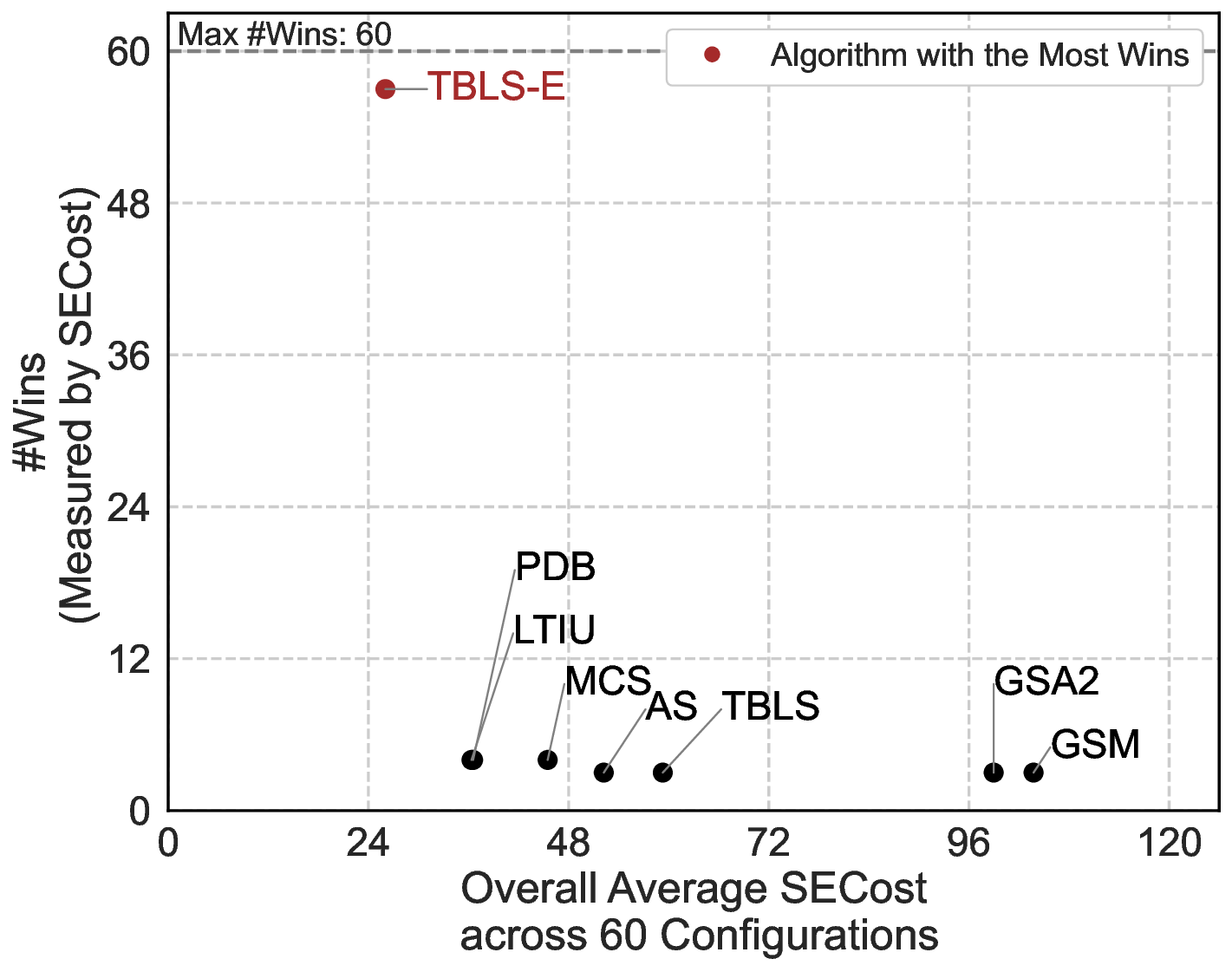}
        \caption{\textnormal{Sex equality cost}}
        \label{fig:small-sized-SMTI-quality-sec}
    \end{subfigure}
    \hfill
    \vspace{-20pt} 
    \caption{\textnormal{Solution quality comparison for small-sized MAX-SMTI problems}}
    \label{fig:small-sized-SMTI-quality}
\end{figure}

Figure \ref{fig:small-sized-SMTI-quality-nos} displays the number of singles produced by various algorithms in solving small-sized MAX-SMTI problems. Our TBLS algorithm achieves the highest number of wins and the lowest average number of singles across 60 configurations. TBLS-E closely follows, also demonstrating excellent performance with an average number of singles very similar to that of TBLS. Furthermore, TBLS-E stands out as the top performer measured by SECost, as illustrated in Figure \ref{fig:small-sized-SMTI-quality-sec}. In almost all configurations, TBLS-E achieves the fairest matchings, significantly outperforming other algorithms. 

Next, we assess the solution quality in addressing large-scale MAX-SMTI problems across 100 configurations, maintaining the same settings for the maximum number of iterations as used for small-sized MAX-SMTI problems. Due to its slow performance, LTIU is excluded from this evaluation. As the instances become harder, the performance gap between algorithms widens. Regarding the number of singles, TBLS remains the winner in most configurations, with approximately 0.2 fewer singles than the second-place TBLS-E and about 1.6 fewer singles than the third-place GSA2, as shown in Figure \ref{fig:large-sized-SMTI-quality-nos}. As illustrated in Figure \ref{fig:large-sized-SMTI-quality-sec}, TBLS-E consistently leads in most configurations, with an average SECost approximately 100 lower than MCS, which ranks second among the algorithms specifically designed for the MAX-SMTI problems.

\begin{figure}
    \centering
    \begin{subfigure}[b]{0.495\textwidth}
        \centering
        \includegraphics[width=\textwidth]{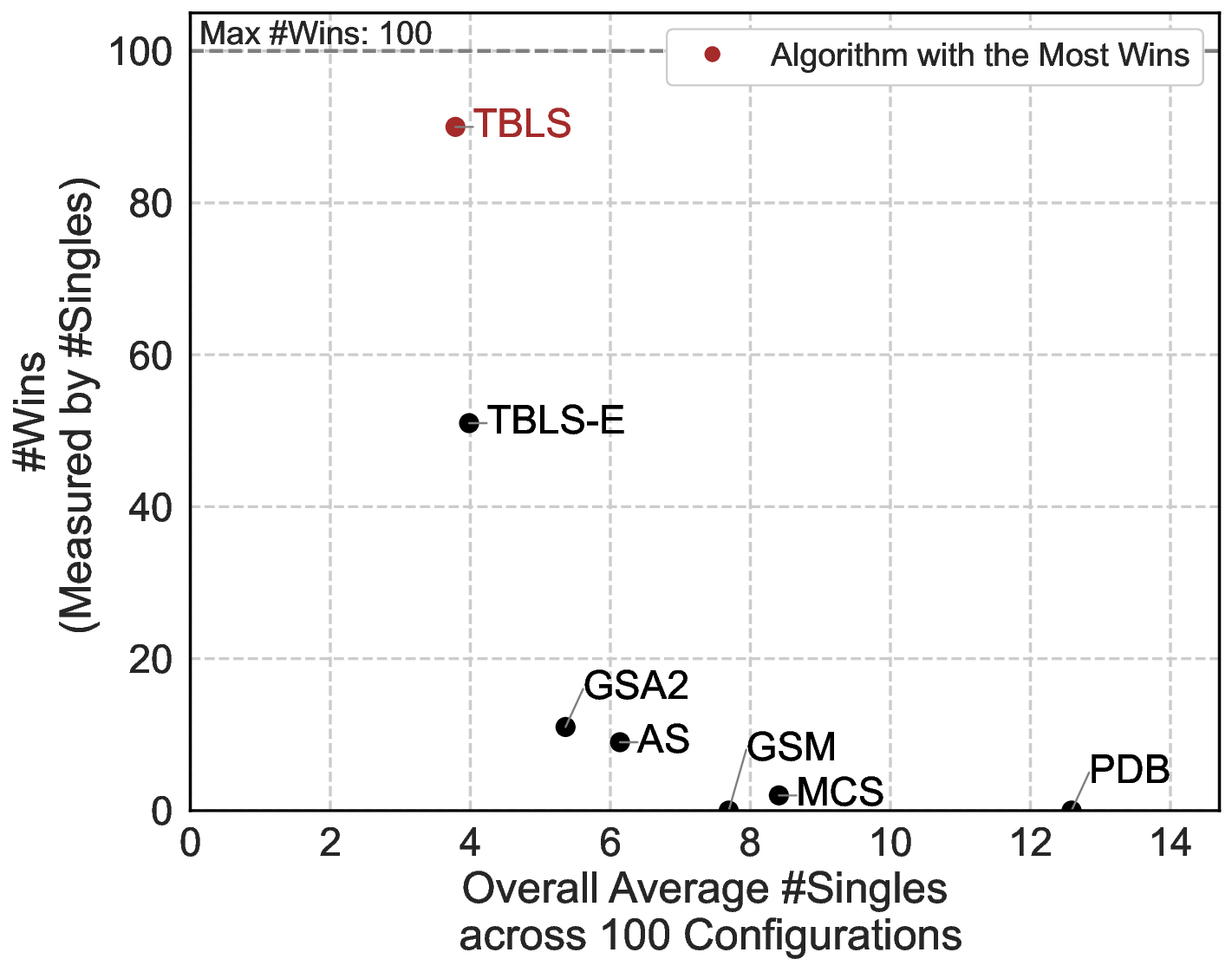}
        \caption{\textnormal{Number of singles}}
        \label{fig:large-sized-SMTI-quality-nos}
    \end{subfigure}
    \hfill
    \begin{subfigure}[b]{0.495\textwidth}
        \centering
        \includegraphics[width=\textwidth]{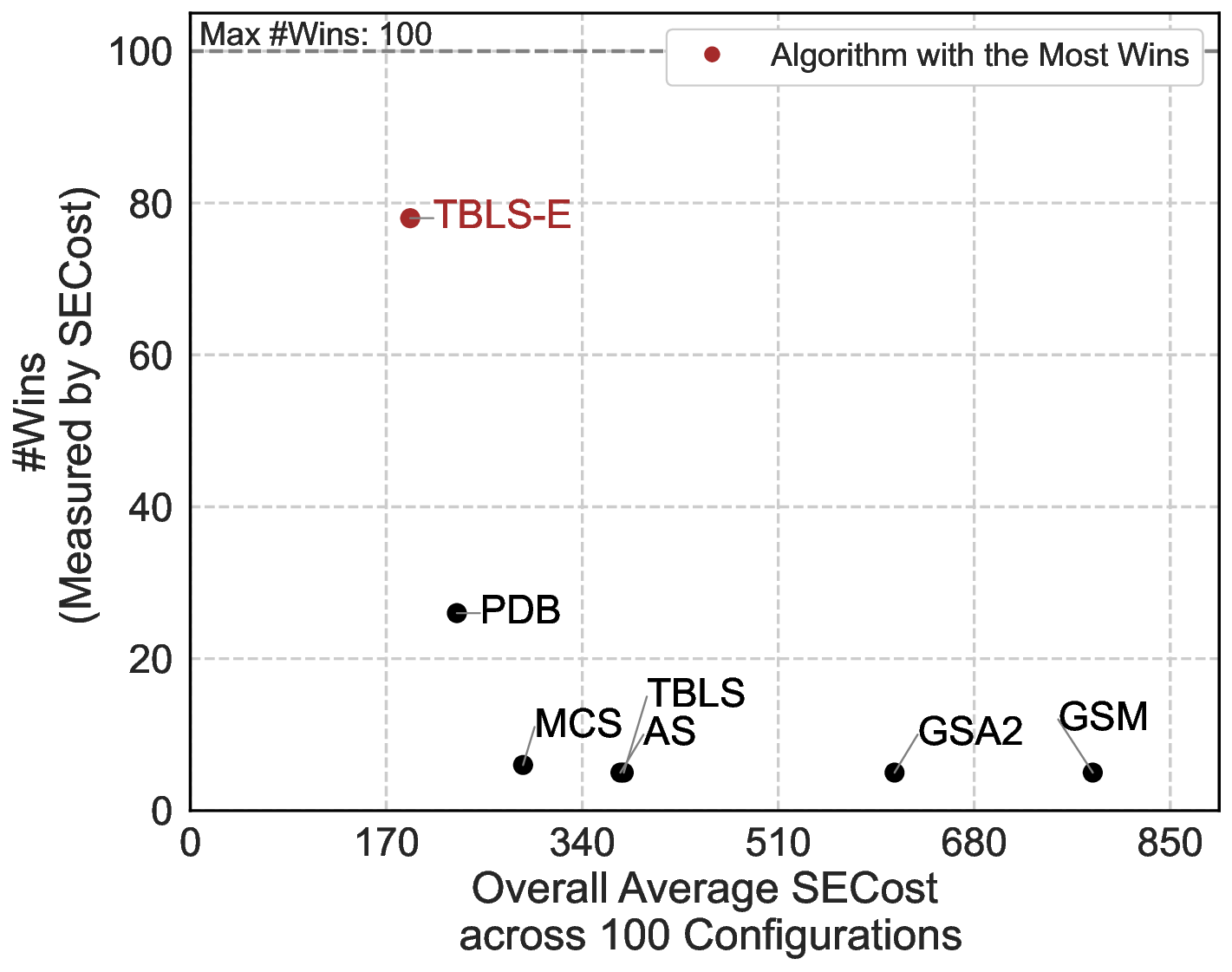}
        \caption{\textnormal{Sex equality cost}}
        \label{fig:large-sized-SMTI-quality-sec}
    \end{subfigure}
    \hfill
    \vspace{-20pt} 
    \caption{\textnormal{Solution quality comparison for large-sized MAX-SMTI problems}}
    \label{fig:large-sized-SMTI-quality}
\end{figure}

\begin{figure}
    \centering
    \begin{subfigure}[b]{0.495\textwidth}
        \centering
        \includegraphics[width=\textwidth]{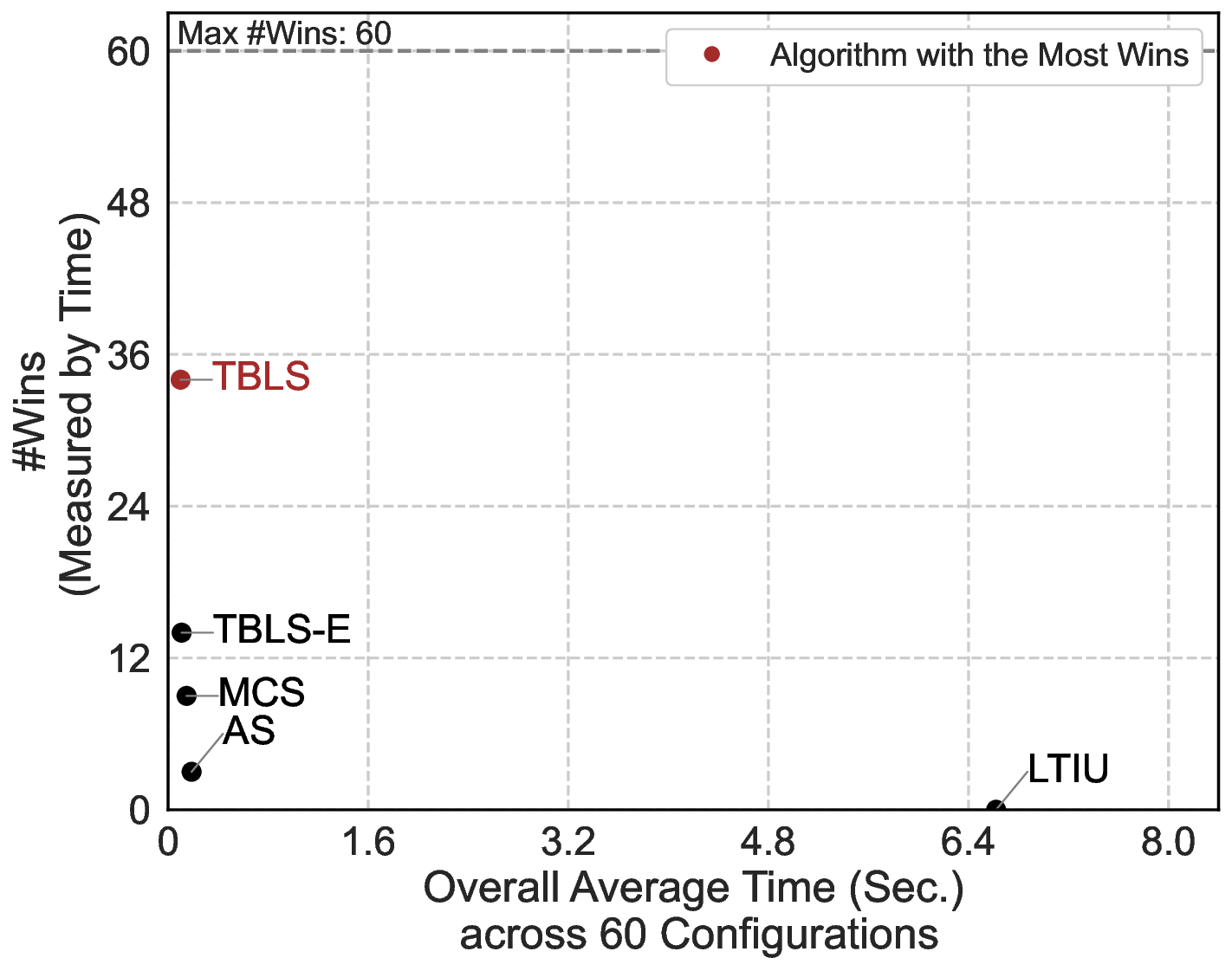}
        \caption{\textnormal{Small-sized MAX-SMTI - execution times}}
        \label{fig:small-sized-SMTI-time}
    \end{subfigure}
    \hfill
    \begin{subfigure}[b]{0.495\textwidth}
        \centering
        \includegraphics[width=\textwidth]{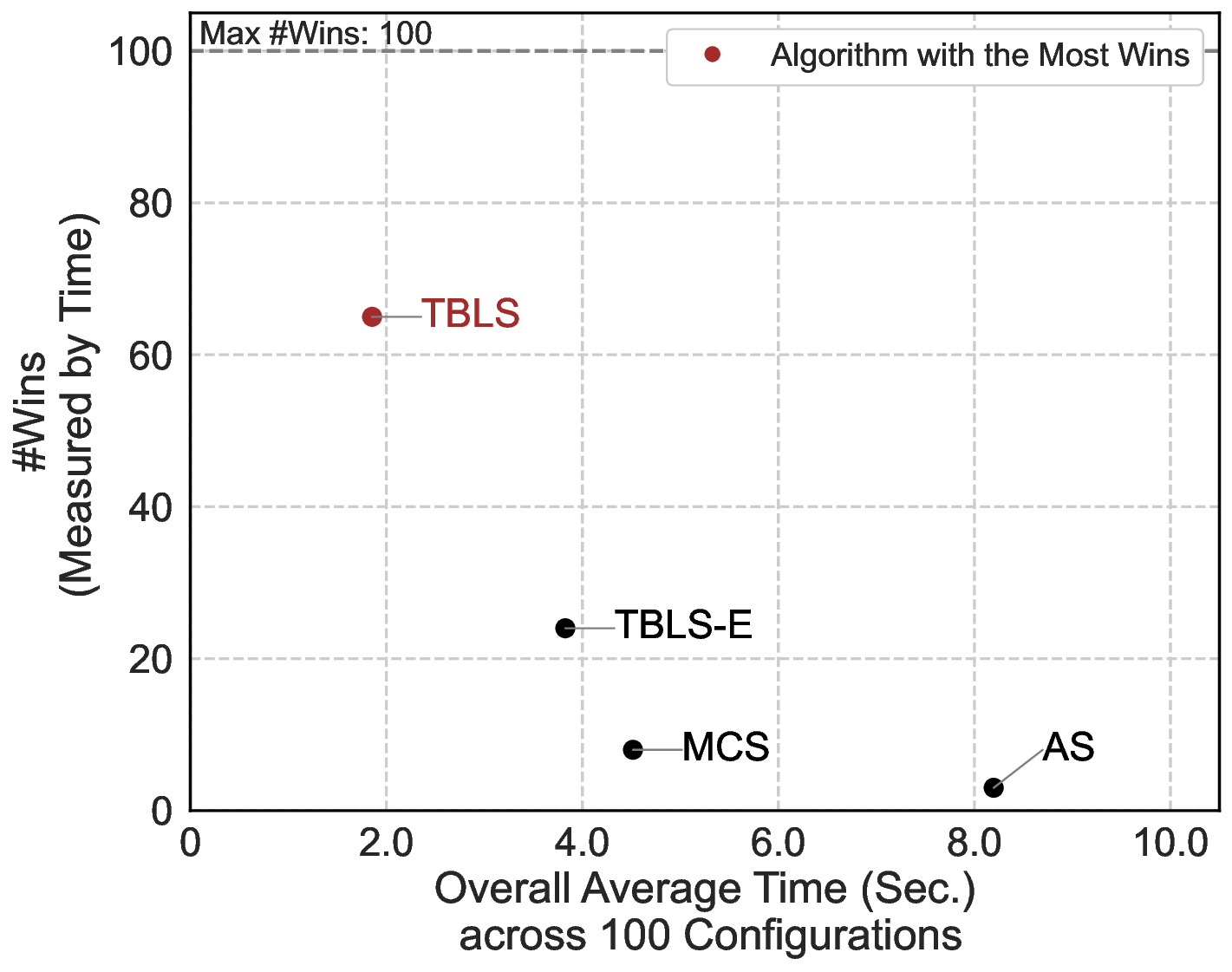}
        \caption{\textnormal{Large-sized MAX-SMTI - execution times}}
        \label{fig:large-sized-SMTI-time}
    \end{subfigure}
    \hfill
    \vspace{-20pt} 
    \caption{\textnormal{Execution time comparison of local search algorithms for MAX-SMTI problems}}
    \label{fig:SMTI-time}
\end{figure}

It is noteworthy that the proposed TBLS-E algorithm outperforms both PDB and all other algorithms, achieving the lowest SECost in 95\% of small-sized configurations and 78\% of large-sized configurations. This improved performance can be attributed to the fact that, in addition to leveraging PDB as its base algorithm, TBLS-E introduces further constraints on tie-breaking adjustments to better balance satisfaction. These results support the effectiveness of TBLS-E in identifying more equitable matchings while maintaining matching sizes comparable to those achieved by its original algorithm TBLS.

Finally, we evaluate the execution times of different local search algorithms in solving small-sized and large-sized MAX-SMTI problems. For these comparisons, the maximum number of iterations is set to 3,000 for all the algorithms. Figure \ref{fig:small-sized-SMTI-time} illustrates the execution times for small-sized MAX-SMTI problems. Although all these algorithms, except LTIU, complete their searches within an average of 0.2 seconds, our TBLS algorithm performs the fastest in most configurations. This gap becomes more pronounced as the problem size increases, which is shown in Figure \ref{fig:large-sized-SMTI-time}. For large-sized MAX-SMTI problems, TBLS and TBLS-E are the top performers in nearly all configurations. On average, the first-place TBLS and the second-place TBLS-E run about 2.5 and 1.2 times faster than the third-place MCS, respectively.

\subsubsection{MAX-HRT problems}\label{subsubsec:HRT}

We begin by evaluating the solution quality of different algorithms for solving MAX-HRT problems of small and large sizes. The maximum number of iterations is still set to 3,000 for TBLS. To maximize the performance of other local search algorithms, the maximum number of iterations is set to 5,000 for both HR and MCA.

Figure \ref{fig:small-sized-HRT-quality-nos} and Figure \ref{fig:large-sized-HRT-quality-nos} illustrate the number of unassigned positions produced by various algorithms in addressing small-sized and large-sized MAX-HRT problems, respectively. Overall, TBLS outperforms all other algorithms in both problem scales, consistently yielding matchings with the fewest unassigned positions. Notably, for large-sized MAX-HRT problems, TBLS produces approximately 1.8 fewer unassigned positions than the second-place algorithm, HPA, and about 11.8 fewer unassigned positions than HR, the second-best among all local search algorithms.

\begin{figure}
    \centering
    \begin{subfigure}[b]{0.495\textwidth}
        \centering
        \includegraphics[width=\textwidth]{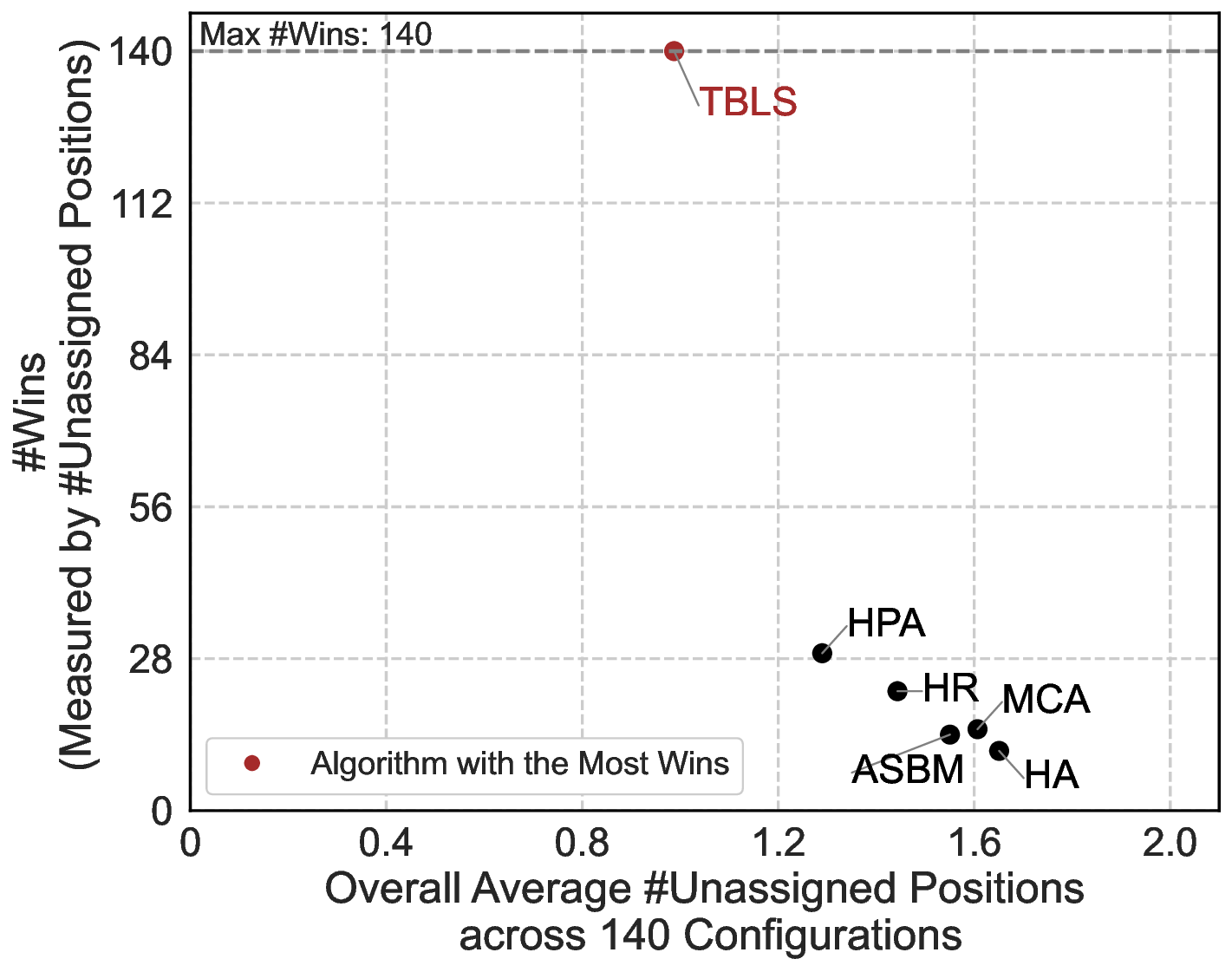}
        \caption{\textnormal{Small-sized MAX-HRT - number of unassigned positions}}
        \label{fig:small-sized-HRT-quality-nos}
    \end{subfigure}
    \hfill
    \begin{subfigure}[b]{0.495\textwidth}
        \centering
        \includegraphics[width=\textwidth]{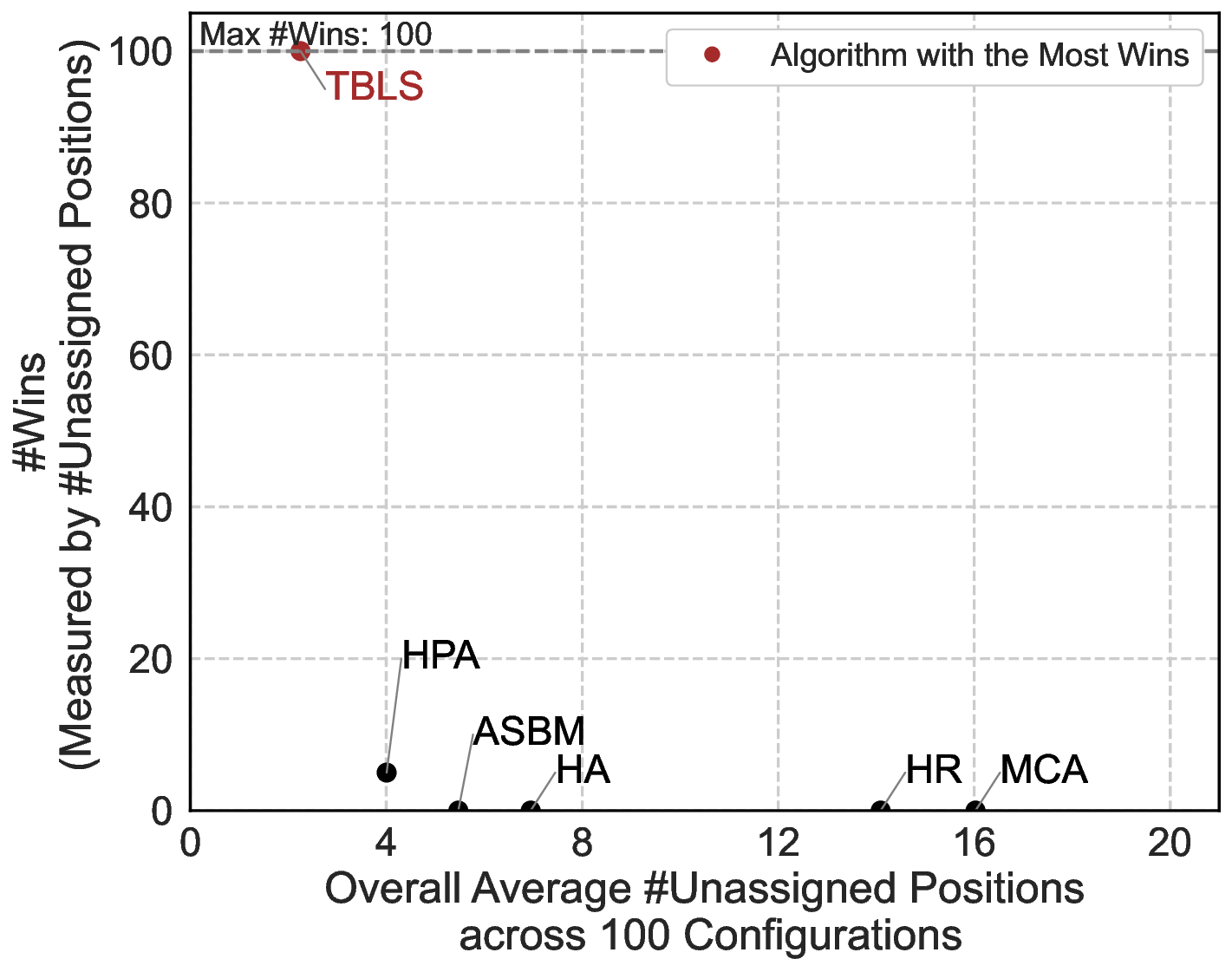}
        \caption{\textnormal{Large-sized MAX-HRT - number of unassigned positions}}
        \label{fig:large-sized-HRT-quality-nos}
    \end{subfigure}
    \hfill
    \vspace{-20pt} 
    \caption{\textnormal{Solution quality comparison for MAX-HRT problems}}
    \label{fig:HRT-quality}
\end{figure}

Additionally, the experimental results reveal that existing local search algorithms suffer a notable decline in performance with increasing problem size $n$. This phenomenon is more obvious for the MAX-HRT problem, as illustrated in Figure \ref{fig:HRT-quality}. One possible explanation is that as input size grows, the solution space expands substantially, and local optima become more difficult to escape. In contrast, our proposed algorithms, TBLS and TBLS-E, maintained robust performance on both problems with larger input sizes. This may be explained by their distinct neighborhood structures: TBLS and TBLS-E employ tie-breaking strategies, whereas other local search algorithms rely on removing blocking pairs. These findings suggest that tie-breaking based neighborhood structures facilitate more effective exploration of the solution space, thereby enhancing both solution quality and robustness, particularly in large-scale instances.

\begin{figure}[h]
    \centering
    \begin{subfigure}[b]{0.495\textwidth}
        \centering
        \includegraphics[width=\textwidth]{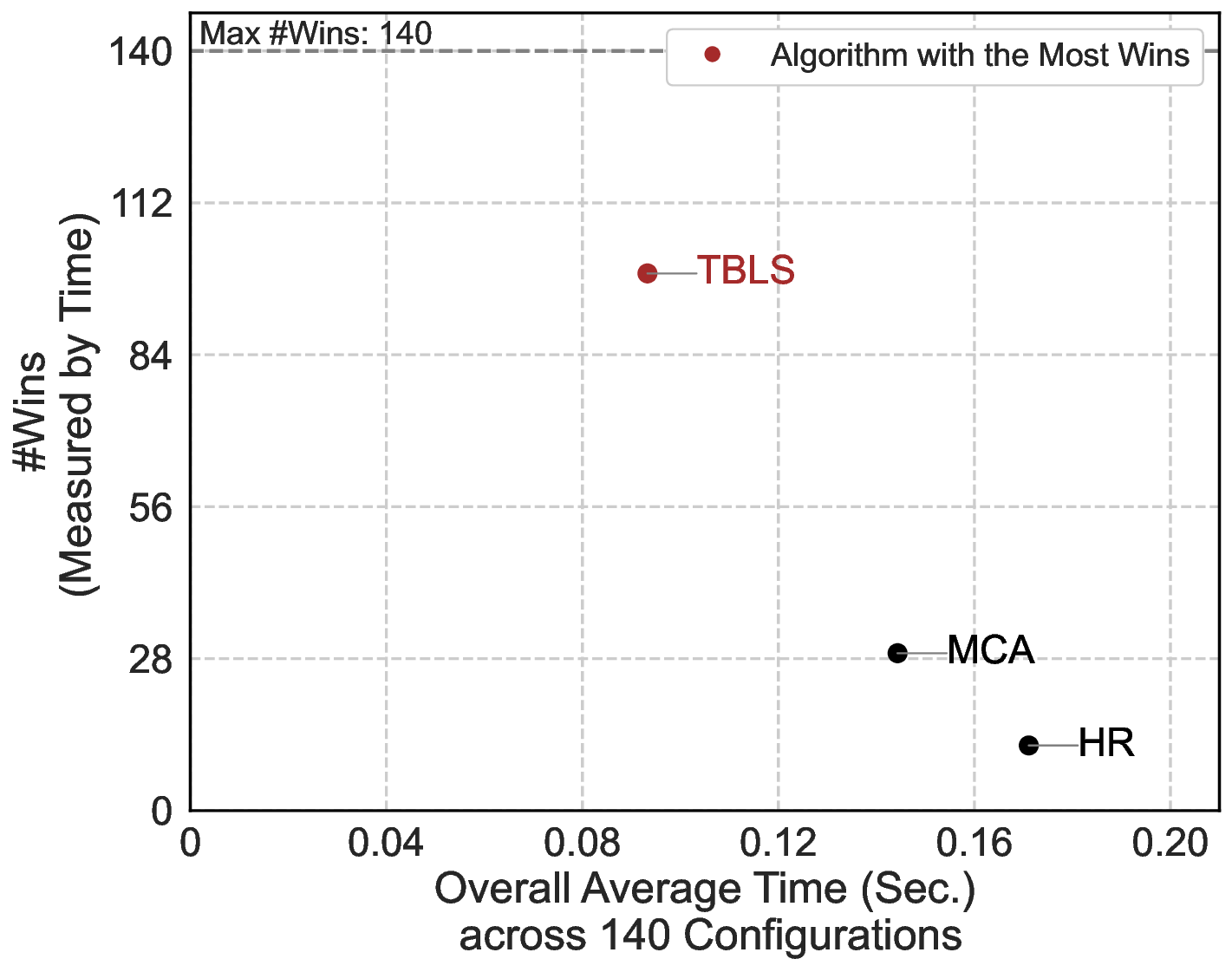}
        \caption{\textnormal{Small-sized MAX-HRT - execution times}}
        \label{fig:small-sized-HRT-time}
    \end{subfigure}
    \hfill
    \begin{subfigure}[b]{0.495\textwidth}
        \centering
        \includegraphics[width=\textwidth]{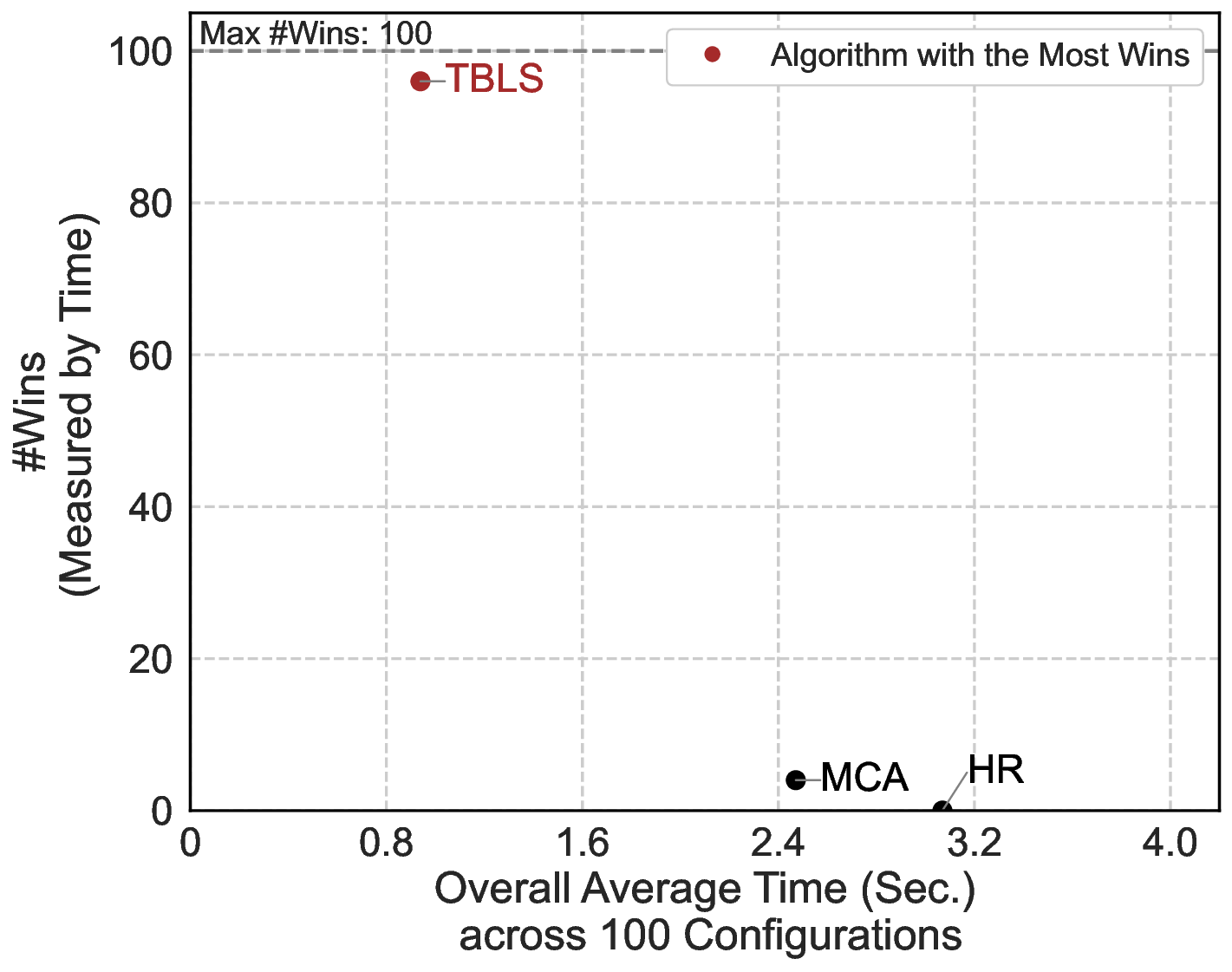}
        \caption{\textnormal{Large-sized MAX-HRT - execution times}}
        \label{fig:large-sized-HRT-time}
    \end{subfigure}
    \hfill
    \vspace{-20pt} 
    \caption{\textnormal{Execution time comparison of local search algorithms for MAX-HRT problems}}
    \label{fig:HRT-time}
\end{figure}

After evaluating the solution quality, we shift our focus to analyzing execution time performance. The maximum number of iterations is set to 3,000 for all the algorithms. Figures \ref{fig:small-sized-HRT-time} and Figure \ref{fig:large-sized-HRT-time} show the execution times for small-sized and large-sized MAX-HRT problems, respectively. In both problem scales, TBLS emerges as the algorithm with the highest number of wins, indicating that it runs the fastest in a large number of configurations. In particular, for large-sized MAX-HRT problems, TBLS operates approximately 2.6 times faster than MCA and about 3.2 times faster than HR.

\subsubsection{Scalability}\label{subsubsec:scalability}

In this section, we increase the input size $n$ from $1,000$ to $10,000$ to evaluate the scalability of our proposed algorithms in comparison with other local search methods. Since this experiment is intended to observe the trend in running time, the maximum number of iterations for each algorithm is set to $3,000$. We use the average number of singles or unassigned positions to assess the solution quality of the algorithm, as defined in Equation (\ref{eq:avg_nos}). If an algorithm fails to produce any stable matching for a given input size, no data point is shown for that algorithm at that size in the figure reporting the average number of singles or unassigned positions.

As shown in Figure \ref{fig:SMTI-scalability-nos}, the AS and MCS algorithms are unable to produce stable matchings for the MAX-SMTI instances within the given number of iterations when the input size exceeds 2,000 and 3,000, respectively. Since our algorithms resolve all ties, they are guaranteed to find a stable matching regardless of the number of iterations or the instance size. Figure \ref{fig:SMTI-scalability-time} depicts the relationship between the input size and the running time for solving MAX-SMTI problems. All four algorithms exhibit growth in time that is slightly faster than linear as the input size increases from 1,000 to 10,000. However, our proposed algorithms appear to have a more moderate growth rate in execution time. Compared with the other two local search algorithms, TBLS and TBLS-E achieve significantly better solution quality and substantially lower running times, demonstrating greater scalability. 

\begin{figure}[h]
    \centering
    \begin{subfigure}[b]{0.495\textwidth}
        \centering
        \includegraphics[width=\textwidth]{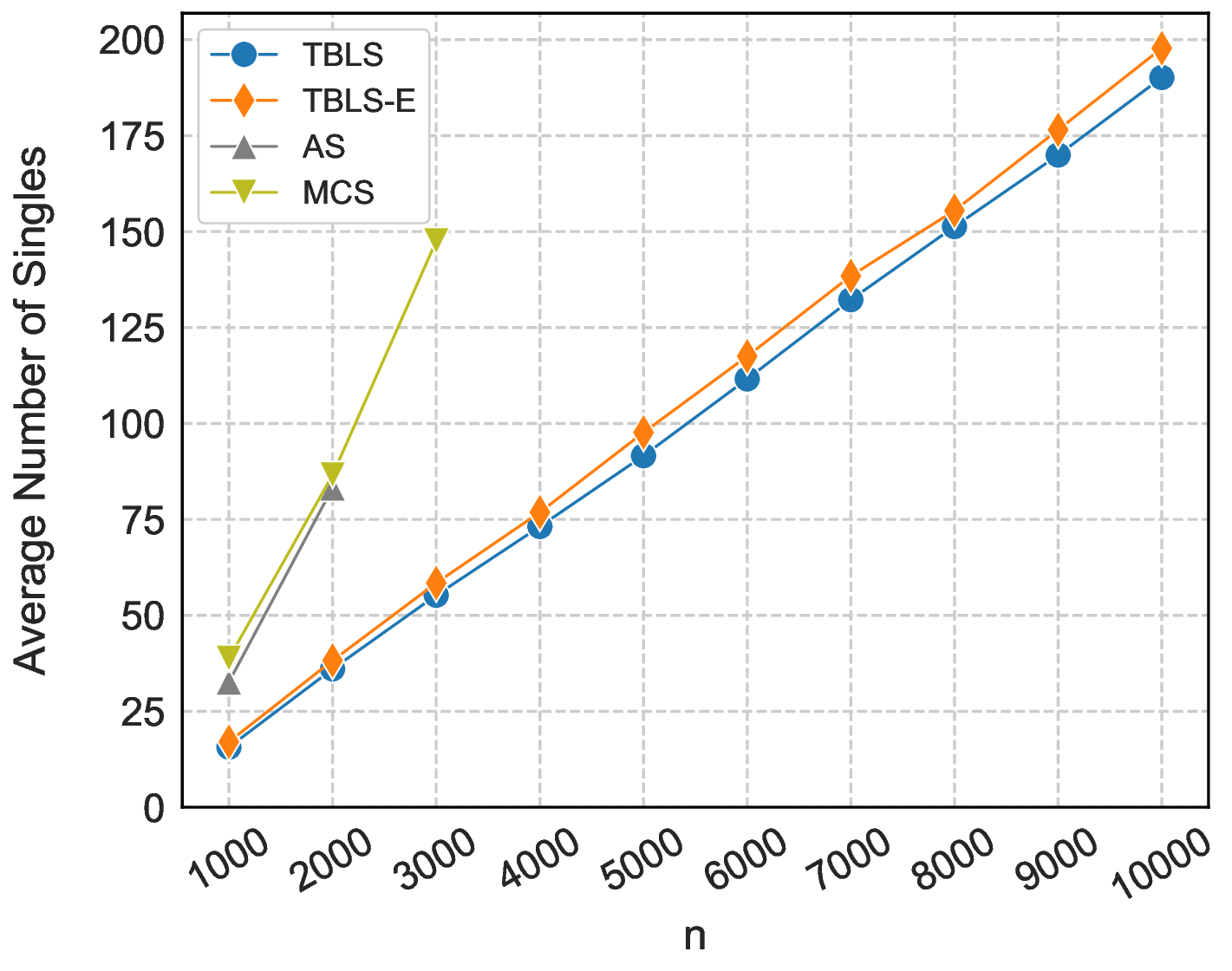}
        \caption{\textnormal{Number of singles}}
        \label{fig:SMTI-scalability-nos}
    \end{subfigure}
    \hfill
    \begin{subfigure}[b]{0.495\textwidth}
        \centering
        \includegraphics[width=\textwidth]{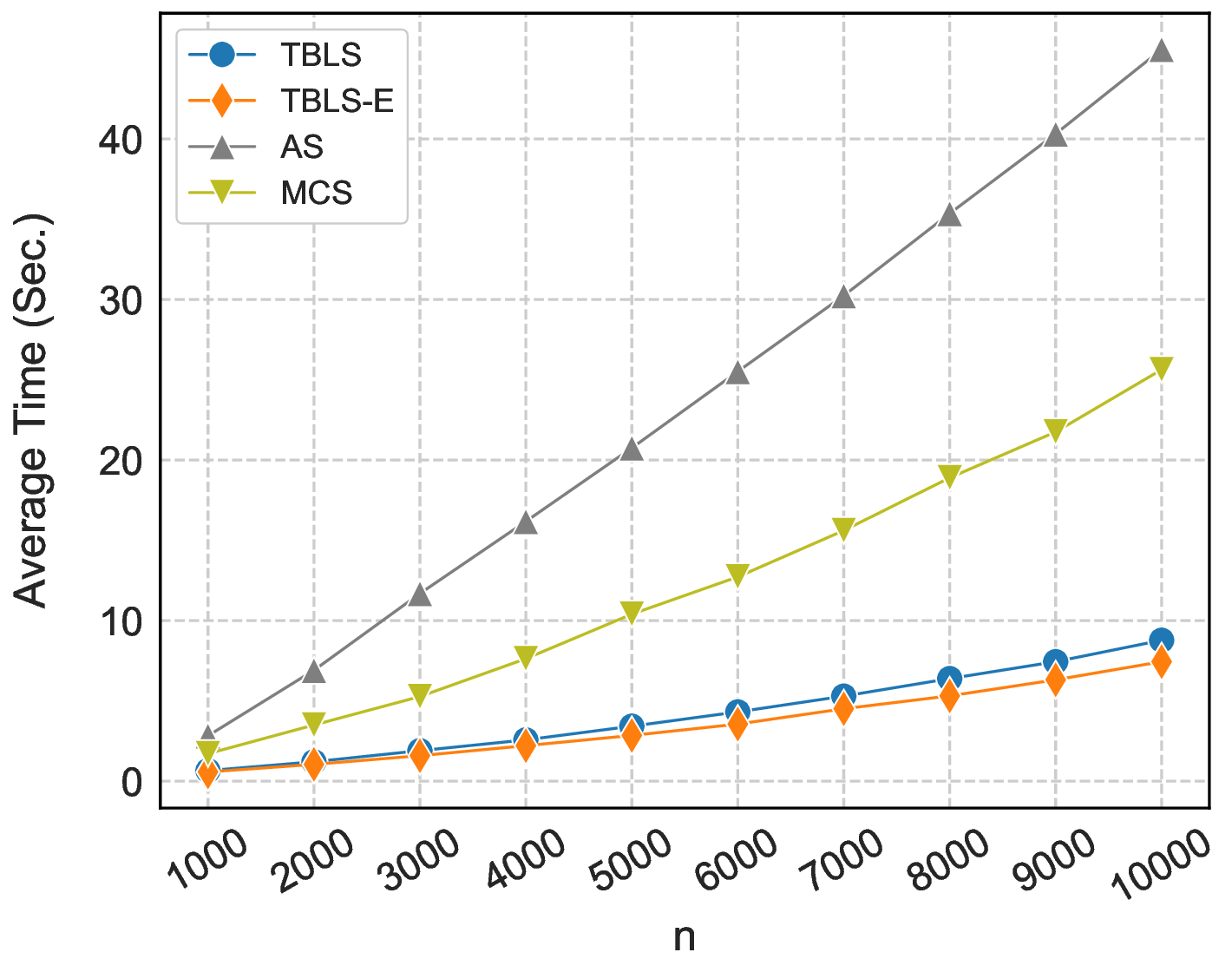}
        \caption{\textnormal{Execution times}}
        \label{fig:SMTI-scalability-time}
    \end{subfigure}
    \hfill
    \vspace{-20pt} 
    \caption{\textnormal{Scalability of four local search algorithms on MAX-SMTI problems}}
    \label{fig:SMTI-scalability}
\end{figure}

\begin{figure}[h]
    \centering
    \begin{subfigure}[b]{0.495\textwidth}
        \centering
        \includegraphics[width=\textwidth]{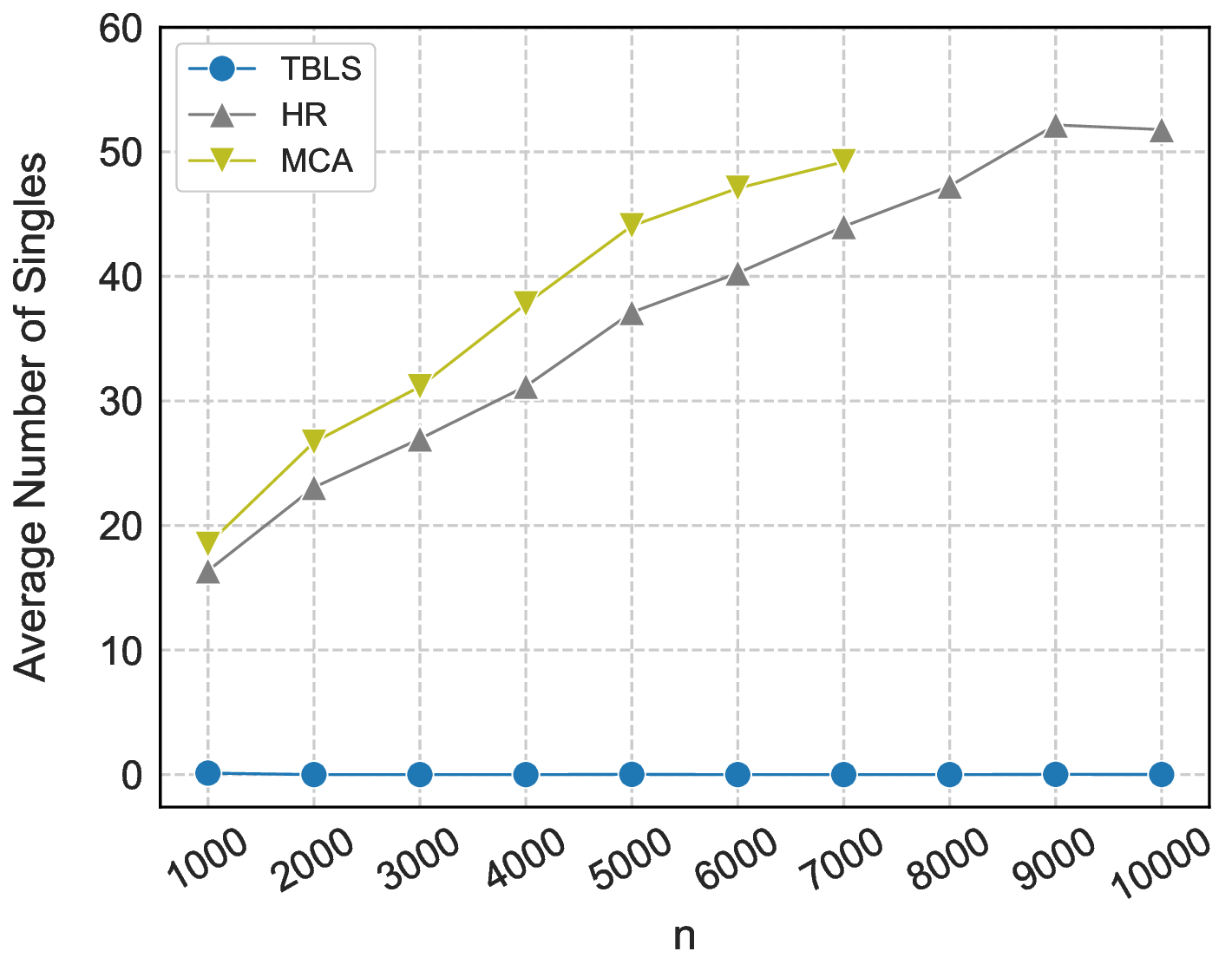}
        \caption{\textnormal{Number of unassigned positions}}
        \label{fig:HRT-scalability-nos}
    \end{subfigure}
    \hfill
    \begin{subfigure}[b]{0.495\textwidth}
        \centering
        \includegraphics[width=\textwidth]{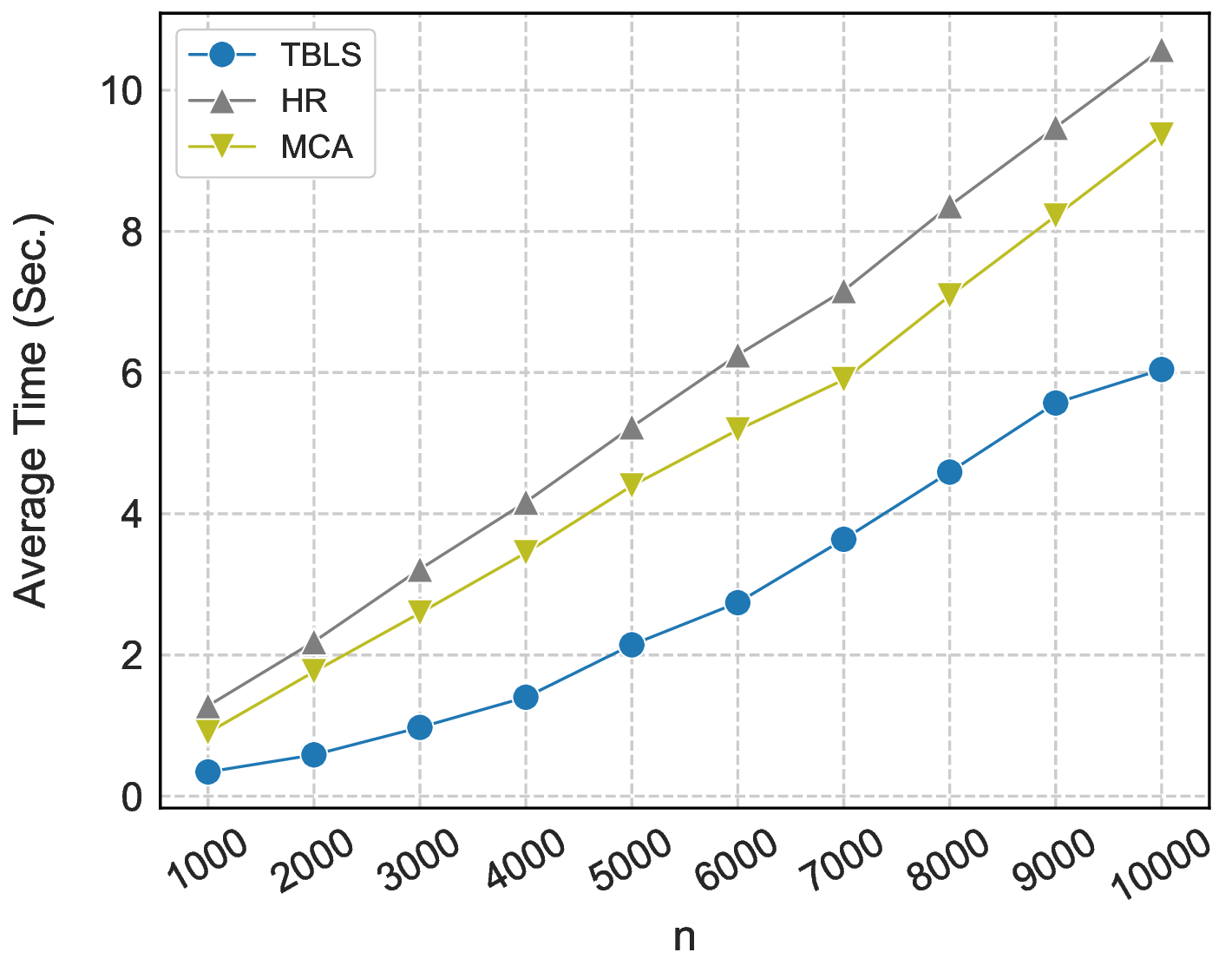}
        \caption{\textnormal{Execution times}}
        \label{fig:HRT-scalability-time}
    \end{subfigure}
    \hfill
    \vspace{-20pt} 
    \caption{\textnormal{Scalability of three local search algorithms on MAX-HRT problems}}
    \label{fig:HRT-scalability}
\end{figure}

Figure \ref{fig:HRT-scalability} compares the scalability of three local search algorithms on instances of the MAX-HRT problem. It is worth noting that TBLS produces perfect matchings for nearly all instances across various input sizes, whereas HR yields an average of 37 unassigned positions, and MCA fails to find stable matchings for instances when $n$ exceeds 7,000. Moreover, as illustrated in Figure \ref{fig:HRT-scalability-time}, TBLS consistently outperforms the other two algorithms in terms of running time. Based on the observed trends over the tested input sizes (from 1,000 to 10,000), there is no evidence suggesting that the growth rate of TBLS's running time will surpass that of the other two algorithms as the input size increases.

\section{Conclusion}\label{sec:conclusion}

This paper proposed a tie-breaking based local search algorithm, named TBLS, to solve the MAX-SMTI and MAX-HRT problems. We also introduced an equity-focused variant, TBLS-E, designed to find relatively fair matchings for the MAX-SMTI problem. Experimental results demonstrate that TBLS consistently outperforms other algorithms in achieving larger matching sizes for both the MAX-SMTI and MAX-HRT problems. Additionally, for the MAX-SMTI problem, TBLS-E yields matchings with lower sex equality costs while preserving matching sizes comparable to those produced by TBLS. Both TBLS and TBLS-E exhibit faster execution times than other local search algorithms when solving large-scale problems. The results of our scalability test suggest that both algorithms continue to perform efficiently with increasing problem size. In the future, we plan to incorporate neural networks and reinforcement learning techniques, such as G3DQN \citep{G3DQN} and NeuRewriter \citep{Neu-Rewirter}, to improve the evaluation function and the refinement process.

\section{Declarations}

\subsection{Funding}
No funding was received for conducting this study.

\subsection{Competing interests}
The author has no competing interests to declare that are relevant to the content of this article.


\end{document}